\documentclass[12pt]{article}

\usepackage{url}
\usepackage{bbm}
\usepackage{mathtools}
\usepackage{amssymb}
\usepackage{amsthm}
\usepackage{empheq}
\usepackage{latexsym}
\usepackage{eurosym}
\usepackage{dsfont}
\usepackage{appendix}
\usepackage{color} 
\usepackage[unicode]{hyperref}
\usepackage{frcursive}
\usepackage[utf8]{inputenc}
\usepackage[T1]{fontenc}
\usepackage{geometry}
\usepackage{multirow}
\usepackage{todonotes}
\usepackage{lmodern}
\usepackage{anyfontsize}
\usepackage{pgfplots}
\usepackage{stmaryrd}
\usepackage{float}
\usepackage{graphicx}
\usepackage{caption}
\usepackage{subcaption}
\usepackage{enumerate}
\usepackage{breakcites}
\graphicspath{ {images/} }

\usepackage{algorithm}
\usepackage{algorithmic}

\pgfplotsset{compat=newest}

\definecolor{red}{rgb}{0.7,0.15,0.15}
\definecolor{green}{rgb}{0,0.5,0}
\definecolor{blue}{rgb}{0,0,0.7}
\hypersetup{colorlinks, linkcolor={red},citecolor={green}, urlcolor={blue}}
            
\makeatletter \@addtoreset{equation}{section}

\newtheorem{theorem}{Theorem}
\newtheorem{theorem2}{Theorem}[section]

\newtheorem*{assumption_sig*}{Assumption $\Ac_\sigma$}
\newtheorem*{assumption_lamb*}{Assumption $\Ac^\lambda$}
\newtheorem*{assumption_max*}{Assumption $\Ac^\text{max}$}
\newtheorem*{assumption_fix*}{Assumption $\Ac^\text{fix}$}
\newtheorem{corollary}[theorem2]{Corollary}

\newtheorem{lemma}[theorem2]{Lemma}
\newtheorem{proposition}[theorem2]{Proposition}

\newtheorem{definition}[theorem2]{Definition}
\newtheorem{remark}[theorem2]{Remark}
\newcommand{\comment}[1]{}
\DeclareUnicodeCharacter{014D}{\=o}
\def\Ac{{\cal A}}

\title{Optimal Rebate Design: Incentives, Competition and Efficiency in Auction Markets}
\author{Thibaut {\sc Mastrolia}\footnote{UC Berkeley, Department of Industrial Engineering and Operations Research, mastrolia@berkeley.edu} ~and~ Tianrui {\sc Xu}\footnote{UC Berkeley, Department of Mathematics, tianrui.xu@berkeley.edu}}

\date{\today}

\usepackage{todonotes}

\begin{document}

\pagestyle{headings}	
\newpage
\setcounter{page}{1}
\renewcommand{\thepage}{\arabic{page}}

\captionsetup[figure]{labelfont={bf},labelformat={default},labelsep=period,name={Figure }}	\captionsetup[table]{labelfont={bf},labelformat={default},labelsep=period,name={Table }}
\setlength{\parskip}{0.5em}
	
\maketitle

\begin{abstract}
This study explores the design of an efficient rebate policy in auction markets, focusing on a continuous-time setting with competition among market participants. In this model, a stock exchange collects transaction fees from auction investors executing block trades to buy or sell a risky asset, then redistributes these fees as rebates to competing market makers submitting limit orders. Market makers influence both the price at which the asset trades and their arrival intensity in the auction. We frame this problem as a principal-multi-agent problem and provide necessary and sufficient conditions to characterize the Nash equilibrium among market makers. The exchange’s optimization problem is formulated as a high-dimensional Hamilton-Jacobi-Bellman equation with Poisson jump processes, which we solve using a verification result. To numerically compute the optimal rebate and transaction fee policies, we apply the Deep BSDE method introduced in \cite{henry2017deep,han2020convergence,dayanikli2025machine}. Our results show that optimal transaction fees and rebate structures improve market efficiency by narrowing the spread between the auction clearing price and the asset’s fundamental value, while ensuring a minimal gain for both market makers indexed to the asset’s price on a coexisting limit order book.


\end{abstract}

\section{Introduction}
\label{intro}


Auction markets have diverse applications, ranging from online advertising and stock trading to electricity markets (see \cite{milgrom2019auction} for an overview of these applications). The profound theoretical and practical contributions to auction theory by 2020 Nobel laureates Robert Wilson and Paul Milgrom \cite{milgrom2021auction} highlight the vital role of market design. In their foundational works, \cite{Wilson1977} and \cite{Milgrom1979} studied the conditions under which an auction market yields a final price that closely approximates the true value of the asset being auctioned. Later, they advised the U.S. Federal Communications Commission and contributed to the implementation of simultaneous multiple-round auctions for spectrum license allocation. This paper explores auction market design for stock exchanges, aiming to improve the efficiency of modern financial markets. \\

The central limit order book (CLOB) operates as a continuous double auction, where participants submit their desired prices and quantities to the exchange. The term ``double'' refers to the dual participation of buyers bidding to purchase shares and sellers offering to sell them. A trade occurs the moment a buyer's bid meets or exceeds a seller's ask. In contrast, a batch auction is a trading mechanism where orders are collected over a fixed time interval and executed simultaneously at a single clearing price, rather than being matched continuously in real time. This process aggregates all buy and sell orders submitted during the batch period, determining a uniform price that maximizes the total executed volume based on supply and demand. 

These two mechanisms underpin the primary trading systems in modern markets: the continuous limit order book and the periodic auction system. The CLOB functions in real time, continuously accepting new orders and matching them instantly when a suitable counterparty is available, either at a specified price (limit orders) or at the best available price (market orders). For empirical evidence, price formation models, and stochastic control problems related to limit order books and market making, we refer the reader to \cite{avellaneda2008high,gayduk2018endogenous,cont2010stochastic,gueant2013dealing,foucault1999order,colliard2012trading,baldacci2023mean,barucci2025market} and the monograph \cite{cartea2015algorithmic}.

By contrast, a periodic auction operates through a sequence of repeated batch auctions. When a trader submits an order, an auction phase is initiated and remains open for a fixed duration, ranging from a few milliseconds to several minutes depending on the exchange. During this period, market participants can submit, modify, or cancel orders. At the close of the phase, the exchange computes a clearing price that maximizes the executed volume. All executable orders are then filled at this single clearing price, regardless of their original submission times. A new auction begins once the previous one concludes. Unlike the continuous system, periodic auctions deliberately defer execution to the end of the interval. This approach has gained growing traction among traditional exchanges, with notable implementations including Cboe U.S. Equities Periodic Auctions, Turquoise (operated by the London Stock Exchange Group), and CME Globex.

\subsection{Periodic auctions versus continuous limit order books}

Although the CLOB is the prevailing trading mechanism globally, its continuous nature introduces significant structural challenges. The influential work of \cite{budish2015high} argues that continuous trading inherently creates transient arbitrage opportunities. These fleeting opportunities incentivize market participants to engage in a technological arms race, competing purely for speed to react faster to market updates and execute orders milliseconds ahead of rivals. Budish et al. further argue that this speed competition is economically wasteful, as it does not increase overall market efficiency but rather forces traders to heavily invest resources into latency reduction. Studies by \cite{wah2013latency} and \cite{fs2012} echo these concerns. For instance, \cite{wah2013latency} shows that high-frequency traders, acting as latency arbitrageurs, can widen bid-ask spreads and disrupt market price discovery.

These structural defects raise a fundamental question of whether alternative trading systems can mitigate these issues. \cite{budish2015high} strongly advocates for frequent batch auctions (conducted, for example, every tenth of a second) to eliminate the inherent flaws of the continuous LOB. Building on this, \cite{paul2021optimal} determined the optimal duration of such auctions using empirical data from Euronext, concluding that auction mechanisms often outperform limit order books. Furthermore, \cite{duffie2017size,derchu2024ahead} investigate hybrid markets where continuous trading sessions alternate with batch auctions. Collectively, these studies advocate for batch auctions as a robust solution to pool liquidity, neutralize the mechanical advantages of high-frequency trading, and ultimately enhance market efficiency.

\subsection{Sniping, transaction fees, and rebates: a motivation from DeFi mechanisms}

In this paper, we investigate the specific market design features required to make a periodic auction market efficient, well-operated, and beneficial for all participants. While periodic auctions successfully curb speed-based competition and aggregate liquidity for better price discovery \cite{paul2021optimal}, their baseline design is not without flaws. A toy two-player auction game studied in \cite{derchu2023equilibria} demonstrates that, under imperfect information and without specific incentives, a naive auction market can be highly inefficient, failing to generate any trade. However, \cite{derchu2023equilibria} proves that introducing transaction fees indexed on each player's half-spread stimulates transactions, and provides a quantitative framework for the optimal fees an exchange should set to attract liquidity. Therefore, rebate and fee policies are fundamental to generating trade in auction markets with imperfect information. 

Furthermore, \cite{mastrolia2024clearing} highlights a critical vulnerability in periodic auctions: unregulated strategic traders have no incentive to reveal their intentions early, preferring to submit orders at the very last moment. This ``sniping'' behavior distorts the auction, increasing the distance between the clearing price and the fundamental efficient price of the asset. To counteract this effect, exchanges can either randomize the clearing time or optimally calibrate fees to penalize snipers. Consequently, designing optimal financial incentives is essential to monitor and regulate the behavior of strategic traders. \textit{``The majority of auctions worldwide require cash bids. Yet in many auctions, bidders do not have cash equal to the sum they wish to bid. As a result, bidders finance part of their bids. This financing may come from the financial markets or from the seller.''} \cite{rhodes2005financing}. 
In continuous limit order book models, incentive structures and make-take fee policies have been thoroughly investigated (see \cite{euch2021optimal}). The core innovation of our work is to efficiently design rebate policies specifically tailored for periodic auction markets. These rebates are offered by the exchange to auction market makers, while standard transaction fees are collected from buyer and seller investors.\\

The design of optimal rebates in auctions is also highly relevant for decentralized finance, particularly for next-generation Decentralized Exchanges combatting ``Just-In-Time'' liquidity sniping. Many DEXs now utilize batch auctions to group user orders over short time intervals (e.g., 30 to 60 seconds), executing them simultaneously at a uniform clearing price. By leveraging custom liquidity pools (such as Uniswap v4's smart contract ``hooks'') and implementing dynamic fee structures, these platforms can effectively mitigate sniping.

\paragraph{Contributions and model relevancy.}
We aim to determine the optimal rebate and transaction fee scheme that an exchange should propose to auction participants to maximize both market efficiency and the exchange's own profitability. We consider an ``English auction'' type framework, characterized by open and observable bids over time. We emphasize that the exchange has a dual objective: improving the quality of the market while ensuring its own economic viability through collected fees. As a critical byproduct, this optimally designed scheme inherently mitigates last-minute sniping behavior. This dual mandate is crucial because, as highlighted in \cite{budish2024stock_exchange}, exchanges have little motivation to migrate trading activity from continuous markets to auction markets unless the new mechanism generates sufficient financial benefits for the venue itself. 

The realism of our mathematical model is grounded in several empirical economic features:
\begin{itemize}
    \item \textbf{Competition between strategic traders (market makers):} We consider strategic market makers who compete within a batch auction. They dynamically adjust their quotes and volumes as limit orders throughout the auction's duration, optimally controlling their trading speed. In return for providing liquidity, these traders receive rebates from the exchange.
    \item \textbf{Stochastic order flow from investors:} Buyer and seller investors submit market orders to trade a specific volume at the final clearing price determined by the exchange, regardless of what that price may be. These investors are not strategic but act spontaneously, paying a fixed transaction fee per share. Crucially, we incorporate the realistic feature that investors may cancel their orders at random times.
    \item \textbf{Heterogeneity of traders:} Market participants are not assumed to be symmetric. We account for varying trading volumes and different execution speeds across different traders.
    \item \textbf{Optimal rebate design and sniper mitigation:} This study provides a direct, quantitative mechanism to monitor and mitigate sniping activity by incentivizing traders to reveal their orders earlier in the auction process. The main advancement of this work over previous studies (e.g., \cite{mastrolia2024clearing}) is the rigorous derivation of an optimal, continuous-time incentive structure. Unlike earlier discrete-time approximations, this paper models time continuously and formulates the interaction as a Nash equilibrium game between strategic traders.
    \item \textbf{Deep learning methods for high-dimensional stochastic control:} Due to the mathematical complexity of the model and the high dimensionality of the resulting multi-player game, we employ a deep learning Backward Stochastic Differential Equation (Deep BSDE) method. This allows us to numerically solve the highly non-linear optimization problem associated with controlled multi-dimensional SDEs with jumps and controlled intensities.
\end{itemize}

\bigskip

To derive this optimal incentive structure, we adopt a principal-agent framework. The stock exchange acts as the principal, and the strategic traders act as the agents. The exchange proposes a set of financial incentives (rebates and fees), and the agents subsequently decide when and at what price to submit their orders, aiming to maximize their own expected utility under the proposed contract. The continuous-time principal-agent problem is deeply rooted in contract theory. The seminal work of \cite{holmstrom1987aggregation} pioneered the continuous-time approach, while \cite{sannikov2008continuous} introduced a tractable framework allowing a principal to offer continuous incentives over an infinite horizon. Building upon this foundation, \cite{cvitanic2018dynamic} provided a comprehensive mathematical framework for solving general continuous-time principal-agent problems, which serves as the theoretical backbone of our derivation.

This framework has been naturally extended to address various financial complexities, including volatility control \cite{cvitanic2018dynamic,chiusolo2024new}, multi-principal competition \cite{mastrolia2018principal,hu2023principal,euch2021optimal}, and multi-agent systems \cite{keun2008optimal,espinosa2015optimal}. In this paper, we focus on a multi-agent setup tailored for optimal market making. We draw upon the methodologies of \cite{elie2019contracting,hernandez2024principal,baldacci2023mean}, which elegantly resolve multi-agent contract problems by linking the principal's optimal incentive design to the existence of a Nash equilibrium among the competing agents. While optimal regulatory policies and make-take fees have been explored for continuous limit order books and dark pools \cite{euch2021optimal,baldacci2021optimal,larsson2023optimal}, our work specifically extends this literature into the realm of periodic batch auctions. We formalize a dynamic make-take fee policy operating within a single batch auction (for instance, a closing auction) where continuous quoting occurs prior to the clearing phase.

\bigskip

A known challenge of applying the continuous-time principal-agent framework to real-world market microstructure is that the optimal incentive scheme is implicitly characterized by a complex, high-dimensional system of partial differential equations. We employ a numerical approach based on the Deep BSDE method to solve for and visualize the optimal rebate structure. This method \cite{henry2017deep,han2020convergence,ji2020three,ji2022deep,dayanikli2025machine} uses deep neural networks to efficiently solve stochastic optimization problems that suffer from the curse of dimensionality. Recent advances in scientific machine learning, coupled with powerful open-source computational tools, have made it highly effective to tackle these previously intractable equations (see \cite{weinan2021algorithms,beck2020overview,raissi2024forward,han2017deep}). Furthermore, our approach aligns with recent successful applications of machine learning in solving stochastic control problems involving jump processes \cite{voskamp2023deep,lu2024multi,georgoulis2024deep}. Finally, this work directly addresses a major frontier in modern economic and financial mathematics, as emphasized by René Carmona. In \cite[Section 5.1]{carmona2022influence}, Carmona notes: \textit{``Given the ubiquitous nature of auctions in the financial industry, we find it rather disheartening that the mathematical finance community did not pay more attention to the large number of applications in search for a deeper theoretical understanding. Computing equilibria, especially Nash equilibria, is a challenging problem...''}. Our research embraces this challenge. By rigorously combining price discovery mechanisms, contract theory, and Nash equilibria in continuous time, this paper provides a quantitative and mathematical response to Carmona's challenge. 

\bigskip

The remainder of this paper is organized as follows. Section \ref{sec:model} introduces the mathematical model, outlining the characteristics of market participants and the rules governing the auction's clearing price. Section \ref{sec:principalagent} formulates the principal-agent framework. Within this framework, Section \ref{sec:agentprb} details the strategic agents' optimization problem, Section \ref{sec:nash} establishes the Nash equilibrium driving the competition among agents, and Section \ref{sec:principalprb} focuses on the exchange's problem of designing the optimal rebate and transaction fee policies. Finally, Section \ref{sec:neural} details the application of the Deep BSDE method to numerically compute and illustrate the optimal market design. The proofs of the technical lemmas are deferred to the appendix.

\section{The auction market model}
\label{sec:model}

\subsection{The auction market}
In this section, we define informally the auction market we are considering. The rigorous mathematical formulation is presented in the subsequent section. We consider a financial asset with an efficient price $P^*$ traded in a batch auction of fixed duration $T>0$. Within this auction, we model two types of participants: strategic market makers, who continuously control their bidding prices and arrival intensities throughout the auction, and institutional investors, who send market orders to buy or sell fixed volumes of the asset.\newline 

The market trades a risky asset whose efficient price $P^*$ is driven by the following diffusion: 
$$d P_t^* = \sigma dW_t,\; P_0^*>0, \quad t \in [0,T],$$ 
where $W$ is a standard Brownian motion and $\sigma>0$ is a fixed volatility constant. Let $P_T^{cl}$ denote the clearing price of the auction determined by the stock exchange at time $T$. The efficient price $P^*$ can be interpreted as the fair no-arbitrage value of the asset, or more realistically, the mid-price of a co-existing limit order book. It acts as the fundamental benchmark from which players aim to deviate in order to generate profit during the auction.
\vspace{0.5em}

\subsubsection{Strategic traders: market makers and supply schedules}
We focus on two strategic traders in this auction, denoted as market maker $p$ and market maker $q$. They submit orders according to right-continuous and adapted counting processes $N^p$ and $N^q$, with respective intensity processes $\lambda^p$ and $\lambda^q$. Here, $N^i_t$ represents the total number of orders submitted by trader $i \in \{p,q\}$ up to time $t$. For instance, a jump $\Delta N^p_t = N^p_t - N^p_{t^-} = 1$ indicates that market maker $p$ submits an order at exact time $t$. 

We assume that these strategic traders penalize deviations of their intensity process from a fixed target intensity $\lambda_0$. This reference level can be viewed as a benchmark trading strategy, analogous to an Almgren-Chriss execution schedule, which the trader aims to track.\footnote{Note that when the trader's risk aversion parameter is $0$, following this benchmark coincides with a ``minimum impact'' approach, commonly known as a TWAP strategy.} 

Furthermore, the market makers continuously update their bidding prices, denoted by $P^p_t$ and $P^q_t$. These prices are based on idiosyncratic signals and deviate from the fundamental efficient price $P^*$. Specifically, let $P^{*,p}$ and $P^{*,q}$ be two independent copies of $P^*$ (driven by independent Brownian motions $W^p$ and $W^q$). The price quoted by trader $i \in \{p,q\}$ evolves as:
$$dP^i_t = \mu^i_t dt + dP_t^{*,i},$$
where the drift $\mu^i_t$ represents the spread the trader actively controls to favor their position. \vspace{0.5em}

These participants are considered ``strategic'' because they jointly optimize their arrival intensity $\lambda^i$ and their spread $\mu^i$. We refer to them as market makers because their quoted prices $P^p$ and $P^q$ ultimately drive the liquidity and pricing dynamics of the auction.

\begin{remark}[Market fragmentation]
The assumption that traders observe independent, noisy copies of the efficient price $P^*$ is a standard feature in market microstructure modeling, \cite{hasbrouck1995one,ohara2015high} motivated by market fragmentation and information asymmetry. While there is a single theoretical fair value $P^*$, traders interact with the market through different technological interfaces, distinct venues, or via correlated financial instruments. Consequently, the price they act upon at time $t$ is a specific realization subject to local noise, latency, and microstructure frictions.
\end{remark}

The executed volume for a market maker's order depends on the final clearing price $P^{cl}_T$ and is modeled via a linear supply/demand schedule function as introduced in \cite{klemperer1989supply}. Specifically, the volume requested by market maker $i$ is given by the time-independent function $Q^i(x) = K^i(P^i_t - x)$, evaluated at the clearing price $x = P^{cl}_T$. If $P^{cl}_T \leq P^i_t$, we have $Q^i(P^{cl}_T) \geq 0$, meaning the order acts as a buy limit order since its limit price is above the clearing price. Conversely, if $P^{cl}_T > P^i_t$, the order is negative ($Q^i(P^{cl}_T) < 0$) and is executed as a sell limit order. 

At the end of the auction, the total aggregated quantity allocated to trader $i$ is the integral of these schedules over their order arrivals:
$$ \int_0^T K^i(P^i_{t} - P^{cl}_T)dN^i_t = \sum_{0 < t \leq T} K^i(P^i_{t} - P^{cl}_T)(N^i_t - N^i_{t^{-}}). $$

\begin{remark}[Demand/Supply curve submission]
Note that each strategic trader submits a linear supply/demand curve to the exchange, see \cite{klemperer1989supply}. this schedule is built by the exchange regarding the limit orders sent by the participants and will be evaluated at time $T$ with respect to the clearing price to determine the volume of orders exchanges. In other words, the volume $Q_i(P_T^{cl})$ is not submitted at time $t$. Rather, each strategic trader submits a price schedule seen as a demand curve whenever
the Poisson process $N^p$ or $N^q$ jumps. Rather than a standard Cournot competition where players compete strictly on fixed quantities, traders here simultaneously submit price-contingent quantity schedules to maximize their expected profits. 
\end{remark}

We denote by $\alpha^p = (\lambda^p,\mu^p)$ and $\alpha^q = (\lambda^q,\mu^q)$ the respective control processes of the two traders. Consequently, the full strategy profile for each trader is summarized by the tuples $(\alpha^p, N^{p}, P^{p}, P^{*,p}, W^{p})$ and $(\alpha^q, N^{q}, P^{q}, P^{*,q}, W^{q})$.
\subsubsection{Spontaneous investors and order flow}
In addition to strategic market makers, non-strategic buyers and sellers also participate in this auction. They arrive according to right-continuous and adapted counting processes, submitting market orders with fixed volumes to be executed at the end of the auction. Let $N_t^{a}$ denote the number of buyers (asking for liquidity) and $N_t^{b}$ the number of sellers (bidding for liquidity) that have arrived up to time $t$. We assume $N_t^{a}$ and $N_t^{b}$ are independent counting processes with respective intensity processes $\lambda^a$ and $\lambda^b$. To encompass the main empirical facts and the realism of auction markets, we enforce the following assumptions:

\begin{itemize}
\item \textit{Order cancellation.} Non-strategic participants may cancel their orders before the auction closes. Let $\theta: \mathbb{R}^+ \to \mathbb{R}^+$ be a non-increasing function with $\theta(0) \geq 1$, and let $(A_i)_{i \geq 1}, (B_i)_{i \geq 1}$ be independent and identically distributed uniform random variables on $[0,1]$. The $i$-th buyer, arriving at time $\tau_i^a$, ultimately submits an order of size $v_i^a \mathbf{1}_{\{A_i \leq \theta(T - \tau^a_i)\}}$. Symmetrically, the $i$-th seller, arriving at time $\tau_i^b$, submits an order of size $v_i^b \mathbf{1}_{\{B_i \leq \theta(T - \tau^b_i)\}}$. Here, the indicator function models the cancellation option (with orders arriving exactly at $T$ never being cancelled, since $\theta(0) \geq 1$), and the base volumes $v^a_i$ and $v^b_i$ are bounded by strictly positive constants $v^a$ and $v^b$, respectively. 

\item \textit{Market impact and order flow.} We assume that the pricing behavior of the strategic market makers influences the arrival of non-strategic investors. Specifically, higher spreads proposed by the market makers attract more sellers and fewer buyers to the market, whereas lower spreads attract more buyers and fewer sellers. 

\item \textit{Transaction fees.} Non-strategic traders incur a fixed transaction fee $d \geq 0$ per executed order at the end of the auction.
\end{itemize}

To rigorously characterize the market impact of the strategic traders and the deterrent effect of the fee $d$ on non-strategic traders, we parameterize the intensities as follows:
$$ \lambda^a_t = \lambda^a(\mu_t,d) := \lambda_0 i^a(\mu^p_t,\mu^q_t)\iota(d) \quad \text{and} \quad \lambda^b_t = \lambda^b(\mu_t,d) := \lambda_0 i^b(\mu^p_t,\mu^q_t)\iota(d), $$
for a baseline intensity $\lambda_0$, a monotonically non-increasing function $i^a$, a non-decreasing function $i^b$, and a non-increasing fee-impact function $\iota$, normalized such that $i^a(0,0) = i^b(0,0) = \iota(0) = 1$. For instance, the exponential specification see \cite[Section 2.5.]{avellaneda2008high}
$$ \lambda^a_t = \lambda_0 e^{-c_1\mu^p_t - c_2\mu^q_t} e^{-d} \quad \text{and} \quad \lambda^b_t = \lambda_0 e^{c_3\mu^p_t + c_4 \mu^q_t}e^{-d}, $$
for non-negative constants $c_1, c_2, c_3, c_4$, satisfies these structural requirements. \vspace{0.5em}
 
\subsubsection{Clearing Price Rule}

At the end of the auction $T>0$, the exchange determines the clearing price to maximize the executed volume. Given a candidate price $x > 0$, the net volume of limit orders received from the strategic market makers is
\[ LM_T(x) = \int_0^T K^p(P_s^p - x)dN_s^p + \int_0^T K^q(P_s^q - x)dN_s^q, \]
while the net volume of market orders submitted by non-strategic buyers and sellers is
\[ MO_T = \sum_{i=1}^{N_T^a}v_i^a \mathbf{1}_{\{A_i \leq \theta(T-\tau^a_i)\}} - \sum_{j=1}^{N_T^b}v_j^b \mathbf{1}_{\{B_j \leq \theta(T-\tau^b_j)\}}. \]

The clearing price, denoted by $P^{cl}_T$, is set by the stock exchange to maximize the total traded volume. Since the submitted volume of each strategic trader is a strictly decreasing linear function of the clearing price, maximizing the traded volume is equivalent to finding the equilibrium point where aggregate supply meets aggregate demand. Similarly to \cite{du2017optimal} and \cite[Section 2.3]{paul2021optimal}, the market clearing condition is given by $LM_T(P^{cl}_T) + MO_T = 0$, or equivalently:
\begin{align*}
    \sum_{i=1}^{N_T^a} v^a_i \mathbf{1}_{\{A_i \leq \theta(T - \tau^a_i)\}} - \sum_{j=1}^{N_T^b} v^b_j \mathbf{1}_{\{B_j \leq \theta(T - \tau^b_j)\}} &+ \int_0^T K^p(P^{p}_s - P^{cl}_T ) dN^{p}_s \\
    &+ \int_0^T K^q(P^{q}_s - P^{cl}_T ) dN^{q}_s = 0.
\end{align*}

Therefore, exploiting the linearity of the submitted supply and demand curves, we can solve explicitly for the clearing price:
\begin{equation}
    P^{cl}_T = \frac{\sum_{i=1}^{N_T^a} v^a_i \mathbf{1}_{\{A_i \leq \theta(T - \tau^a_i)\}} - \sum_{j=1}^{N_T^b} v^b_j \mathbf{1}_{\{B_j \leq \theta(T - \tau^b_j)\}} + \int_0^T K^p P^{p}_s dN^{p}_s + \int_0^T K^q P^{q}_s dN^{q}_s} {K^p N^{p}_T + K^q N^{q}_T}.
    \label{eq:pclearing}
\end{equation}

\begin{remark}[Existence of the clearing price and terminology Auction makers]\label{rem:existencecp}
Note that the clearing price in Equation \eqref{eq:pclearing} is well-defined almost surely, provided that $K^p N^{p}_T + K^q N^{q}_T > 0$, meaning at least one strategic order has been submitted during the auction. On the negligible event where no strategic orders are submitted, we assume the exchange defaults to a reference price, such as the efficient price $P^*_T$. This is why we will from now on call these strategic actors Auction Makers. We thus assume that the duration of the auction is long enough such that at least an Auction Maker enters in it. Alternatively, we could assume that there always exists at least on limit order in the auction. This is in accordance with the main assumption in \cite{paul2021optimal}.
\end{remark}

\subsubsection{Toy example and motivation}

We propose to consider a simple framework as a toy model to understand the motivation of each players. For the sake of simplicity, we are not considering auction market orders, that is $v^a=v^b=0$. Consider the following schedules sent by Market Maker $p$ with limit price $100$ and Market Maker $q$ with limit price $90$. 

$$\begin{array}{|c|c|}
\hline
\textbf{Volume $MM^p$} & \textbf{Proposed price $P^p$} \\
\hline
-400 & 96 \\
-300 & 97 \\
-200 & 98 \\
-100 & 99 \\
100 & 101 \\
200 & 102 \\
\hline
\end{array}\qquad\begin{array}{|c|c|}
\hline
\textbf{Volume $MM^q$} & \textbf{Proposed price $P^q$} \\
\hline
-400 & 86 \\
-300 & 87 \\
-200 & 88 \\
-100 & 89 \\
100 & 91 \\
200 & 92 \\
\hline
\end{array}$$

The exchange rebuilding the supply/demand schedule get:
\[Q^p(x)= 100(x-100),\; Q^q(x)=100(x-90).\]

We represent this situation graphically as follow (using the absolute value of $Q^p,Q^q$ for the sake of clarity): 
\begin{center}
\includegraphics[scale=0.25]{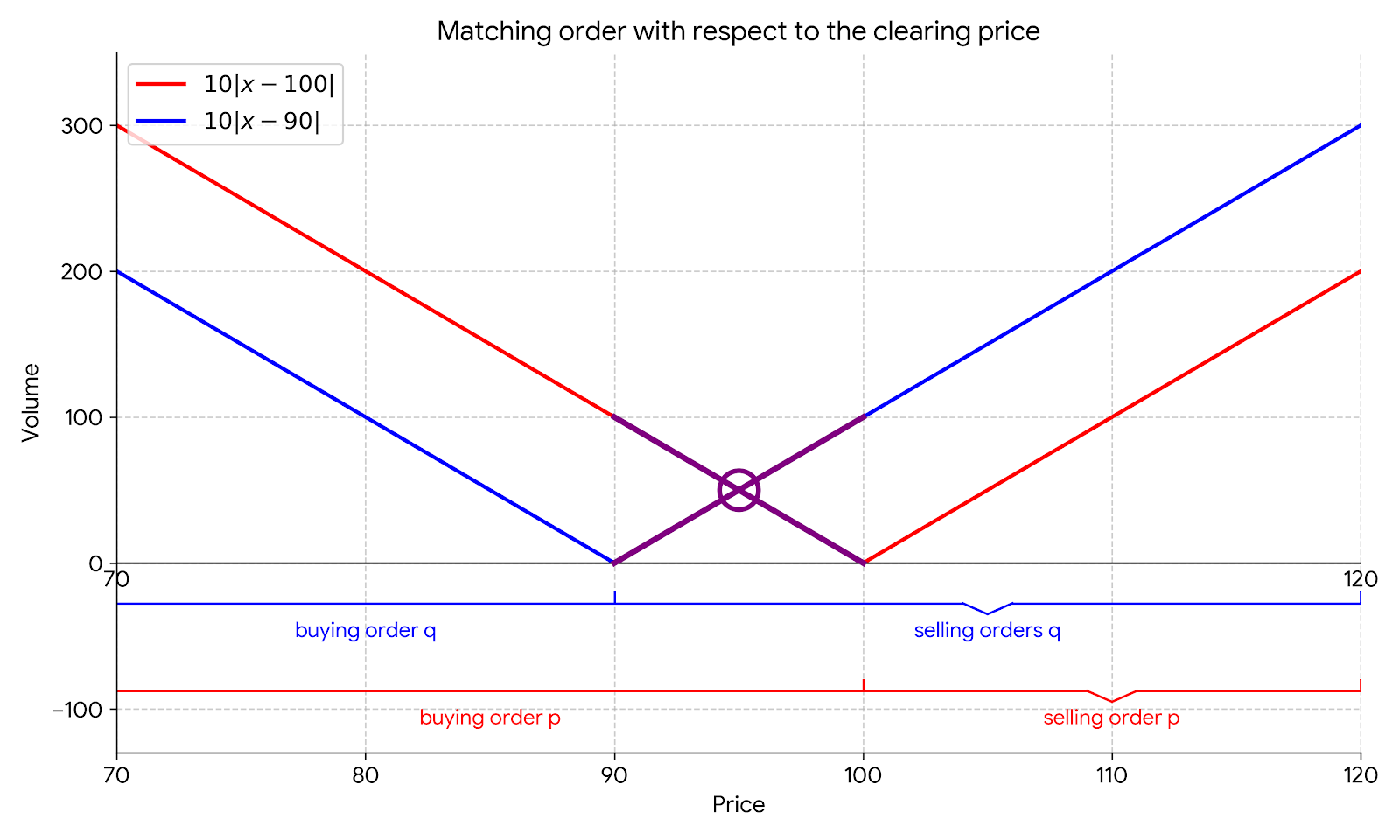}
\end{center}
The clearing price is thus set at $P^{cl}=95$ to ensure as much order executed as possible. Assume that the efficient price market value is $P^*=97.5$. Market Maker $p$ buy at $95$ instead of $97.5$ a quantity of $500$ making a net profit of $500\times 2.5=1250$. At the opposite, Market Maker $q$ sell at $95$ instead of $97.5$ the same quantity, making a loss of the same amount. This symmetric zero sum game makes Market Maker $q$ to deviate closer to the curve of Market Maker $p$. In this case, there is no clearing price as emphasized in \cite{derchu2023equilibria}. Introducing fees cured this flaw as again explained in a one period model \cite{derchu2023equilibria}.
\subsubsection{Rebates policy}

To incentivize liquidity provision, the exchange reallocates the aggregate transaction fees $d$ collected from non-strategic investors to the strategic market makers. These incentives take the form of terminal rebates, denoted by $\xi^p$ and $\xi^q$, which are settled at the clearing time $T$ based on the realized market activity during the auction. While the functional forms of these rebates (the payment rules) are announced by the exchange at time $t=0$, the actual payouts are $\mathcal{F}_T$-measurable random variables, indexed on the final market characteristics at time $T$.

\begin{remark}[Connections to DeFi and Uniswap v4]
The mechanism of dynamically reallocating collected fees to liquidity providers as a function of market activity strongly resonates with recent structural innovations in Decentralized Finance. In particular, this setup mimics the architecture of Automated Market Makers like Uniswap v4, where custom ``hooks'' enable dynamic fee tiers and the redistribution of internalized MEV (Maximal Extractable Value) or transaction fees back to liquidity providers at the end of a block or batch. Therefore, our theoretical framework extends naturally to the study of optimal liquidity provision in modern decentralized exchanges.
\end{remark}

\subsection{Mathematical framework}

We define the canonical space $\Omega = \mathcal{S} \times \mathcal{S} \times C([0,T],\mathbb{R}) \times \mathcal{S} \times C([0,T],\mathbb{R}) \times \mathcal{S} \times C([0,T],\mathbb{R})$, where $\mathcal{S}$ is the space of right-continuous functions with left limits (càdlàg) from $[0,T]$ into $\mathbb{N}$, and $C([0,T],\mathbb{R})$ is the space of continuous functions from $[0,T]$ into $\mathbb{R}$. Let $\mathcal{F}$ be the Borel $\sigma$-algebra of $\Omega$. The canonical process on $(\Omega,\mathcal{F})$ is denoted by 
$C_t = (C^1_t, C^2_t, C^3_t, C^4_t, C^5_t, C^6_t, C^7_t) := (N_t^a , N_t^b, W_t, N_t^{p}, W_t^p, N_t^{q}, W_t^q)$. 
Let the filtration $\mathbb{F} = (\mathcal{F}_t)_{t \in [0,T]}$ be the natural filtration generated by $(C_t)_{t \in [0,T]}$. 

We introduce a reference probability measure $\mathbb{P}^{0}$ under which $W_t$, $W_t^p$, and $W_t^q$ are independent standard Brownian motions, and under which $N_t^{a}$, $N_t^{b}$, $N_t^{p}$, and $N_t^{q}$ are independent Poisson processes with a constant deterministic intensity $\lambda_0>0$. 

\begin{definition}[Auction market]
The space $(\Omega,\mathcal{F},\mathbb{P}^0)$ is called an auction market where:
\begin{itemize}
    \item $N^a$ and $N^b$ model the arrivals of buy and sell orders from non-strategic investors, arriving at a rate $\lambda_0$ under the reference measure $\mathbb{P}^0$;
    \item $\sigma W$ is the random process driving the efficient price of the traded asset, for some volatility constant $\sigma>0$;
    \item $N^p$ and $N^q$ denote the respective order flows submitted by market maker $p$ and market maker $q$, arriving at a rate $\lambda_0$ under $\mathbb{P}^0$, while $\sigma W^p$ and $\sigma W^q$ represent the idiosyncratic noise in their respective bidding prices. 
\end{itemize}
\end{definition}

In order to construct the controlled market probability where intensities and drifts are actively chosen by the participants, we rely on the following technical result, adapted from \cite[Theorems 5.3 and 5.4]{bjork2021point}. The proof is postponed to Appendix \ref{app:girsanov} 

\begin{lemma}[Change of probability measure and Girsanov's Theorem]\label{thm:girsanoc}
Fix a filtered probability space $(\Omega, \mathcal{F}, \mathbb{P}, \mathbb{F})$. Let $N^1, \dots, N^k$ be a sequence of mutually independent optional counting processes with respective intensities $\lambda^1, \dots, \lambda^k$, which are non-negative optional processes. Let $W^1, \dots, W^\ell$ be mutually independent standard Brownian motions, and let $\phi^1, \dots, \phi^k$ be predictable processes such that $\phi_t^i \geq -1$ for any $i = 1, \dots, k$, $\mathbb{P}$-a.s. Furthermore, let $\delta^1, \dots, \delta^\ell$ be $\mathbb{F}$-adapted processes. 
Assume that the following Novikov-type condition holds:
\begin{equation}\label{novikov}
    \mathbb{E}^\mathbb{P}\left[\exp\left(\frac{1}{2} \int_0^T \sum_{i=1}^\ell |\delta^i_t|^2 dt + \sum_{j=1}^k \int_0^T \Big((1+\phi_t^j)\ln(1+\phi_t^j) - \phi_t^j \Big)\lambda^j_t dt\right) \right] < \infty.
\end{equation}

Then:
\begin{itemize}
    \item[i.] There exists a unique strictly positive martingale solution $L$ to the SDE
    $$dL_t = L_t \sum_{i=1}^\ell \delta_t^i dW_t^i + L_{t^-} \sum_{j=1}^k \phi_t^j (dN_t^j - \lambda_t^j dt), \quad L_0 = 1,$$
    such that $\mathbb{P}(L_t > 0) = 1$ for every $t \in [0,T]$. Consequently, there exists a probability measure $\mathbb{Q}$ equivalent to $\mathbb{P}$ defined by the Radon-Nikodym derivative $\frac{d\mathbb{Q}}{d\mathbb{P}}\big|_{\mathcal{F}_t} = L_t$.
    \item[ii.] Under the probability measure $\mathbb{Q}$, $N^1, \dots, N^k$ are counting processes with compensators $\lambda_t^i(1+\phi^i_t)$ for $i = 1, \dots, k$.
    \item[iii.] The processes $\tilde{W}^1, \dots, \tilde{W}^\ell$ defined by $d\tilde{W}_t^i = dW_t^i - \delta_t^i dt$, for $i= 1, \dots, \ell$, are standard Brownian motions under $\mathbb{Q}$.
    \item[iv.] Under the probability measure $\mathbb{Q}$, $N^1, \dots, N^k$ do not have common jumps, i.e., $\mathbb{Q}(\Delta N_t^i \Delta N_t^j = 0) = 1$ for every $t \in [0,T]$ and every $i \ne j \in \{1, \dots, k\}$.
\end{itemize}
\end{lemma}
As a consequence of this lemma, we derive the controlled auction market as follows.

\begin{corollary}[Controlled auction market]
Let $(\Omega,\mathcal{F},\mathbb{P}^0)$ be an auction market, and let $\alpha \in \mathcal{A}$. There exists a controlled probability measure $\mathbb{P}^\alpha \sim \mathbb{P}^0$ and a tuple $(N^a, N^b, P^*, N^p, P^p, N^q, P^q)$ such that:
\[dP^*_t = \sigma dW_t, \quad dP^p_t = \mu^p_t dt + \sigma d\tilde{W}^p_t, \quad \text{and} \quad dP^q_t = \mu^q_t dt + \sigma d\tilde{W}^q_t,\] 
where $W$, $\tilde{W}^p$, and $\tilde{W}^q$ are independent standard Brownian motions under $\mathbb{P}^{\alpha}$. Furthermore, $N_t^{a}$, $N_t^{b}$, $N_t^{p}$, and $N_t^{q}$ are counting processes with respective intensities $\lambda^a_t, \lambda^b_t, \lambda^p_t, \lambda^q_t$ under $\mathbb{P}^{\alpha}$, and they almost surely have no simultaneous jumps. 

The tuple controlled by $\alpha$ under this measure is called the controlled auction market model, and $\mathbb{P}^\alpha$ is the controlled market probability.
\end{corollary} 

\begin{remark}[Integrable spread and bounded trading intensity]
Assuming that $\mu$ and $\lambda$ are uniformly bounded is stronger than the standard Novikov condition required to well-define $L^\alpha$, but it is a convenient assumption to ensure that $L^\alpha$ admits moments of all orders, which simplifies the technical proofs that follow. From a practical standpoint, bounding $\lambda$ is highly realistic since algorithmic traders are inherently limited in their speed of execution by latency and technological constraints. 

If one wishes to model the spread controls as unbounded processes (for instance, driven by unbounded exogenous signals, as in \cite{derchu2023equilibria}), one could instead assume that the relevant stochastic integrals belong to the BMO (Bounded Mean Oscillation) martingale space, which is sufficient to guarantee that $L^\alpha$ preserves moments of all orders (see \cite[Chapter 3]{kazamaki2006continuous}).
\end{remark}

We denote by $\mathbb{E}^0$ the expectation with respect to the reference measure $\mathbb{P}^0$, and by $\mathbb{E}^{\alpha}$ the expectation with respect to the controlled measure $\mathbb{P}^{\alpha}$.
\section{Rebate design: A Principal-Multi-Agent framework}
\label{sec:principalagent}

We now investigate the exchange's problem of designing the rebate policy as a continuous-time Principal-Multi-Agent problem. We consider the stock exchange as the principal and the market participants as the agents, divided into two main categories: non-strategic investors (buyers and sellers) and strategic auction market makers. The exchange aims to design an optimal rebate policy to incentivize the strategic auction makers, while simultaneously setting a uniform transaction fee $d$ for the non-strategic investors. 

To formalize the contracts, we first introduce the market characteristics process as a tuple $X_t = (X^1_t, X^2_t, X^3_t, X^{4}_t, X^{5}_t, X^{6}_t, X^{7}_t)$, whose components represent the observable state variables of the auction up to time $t$:
\begin{itemize}
    \item \textbf{Non-strategic order volumes:} 
    $X^1_t := \sum_{i=1}^{N_t^a} v^a_i \mathbf{1}_{\{A_i \leq \theta(T - \tau^a_i)\}}$ and $X^2_t := \sum_{j=1}^{N_t^b} v^b_j \mathbf{1}_{\{B_j \leq \theta(T - \tau^b_j)\}}$ denote the aggregate volumes of non-cancelled buy and sell orders, respectively.
    \item \textbf{Reference price:} 
    $X^3_t := P_t^{*}$ is the efficient price of the asset.
    \item \textbf{Auction Maker $p$'s activity:} 
    $X^4_t := \int_0^t P^p_s dN^p_s$ is the cumulative sum of quoted prices, and $X^5_t := N_t^p$ is the total number of orders submitted by market maker $p$.
    \item \textbf{Auction Maker $q$'s activity:} 
    $X^6_t := \int_0^t P^{q}_s dN^{q}_s$ is the cumulative sum of quoted prices, and $X^7_t := N_t^q$ is the total number of orders submitted by market maker $q$.
\end{itemize}
For notational convenience, we denote by $P = (P^p, P^q)$ the bid price processes controlled by both market makers.\newline 
A \textbf{rebate contract} (or incentive scheme) $\xi = (\xi^p,\xi^q)$ is a pair of $\mathcal{F}_T$-measurable random variables representing the terminal payments offered by the exchange to Auction Maker $p$ and Auction Maker $q$, respectively. 

We define the set of admissible incentives $\mathcal{C}^{R_0}$ as the space of contracts that are integrable under any admissible control, and that satisfy the standard \textit{Individual Rationality (IR)} (or participation) constraint for both agents:
\begin{align*}
\mathcal{C}^{R_0} = \Bigg\{ \xi = (\xi^p, \xi^q): \;& \mathcal{F}_T\text{-measurable}, \; \sup_{\alpha \in \mathcal{A}} \mathbb{E}^{\alpha}[|\xi^p|^2] < \infty, \; \sup_{\alpha \in \mathcal{A}} \mathbb{E}^{\alpha}[|\xi^q|^2] < \infty, \\
& V^i_0(\xi^i, d, \alpha^j; \alpha^i) \geq R_0^i, \quad \text{for } i \neq j \in \{p,q\} \Bigg\},
\end{align*}
where $V^i_0(\xi^i,d,\alpha^j; \alpha^i)$ denotes the value function (expected terminal PnL) of Auction Maker $i$, which is rigorously defined in the next section, given the competitor's strategy $\alpha^j$ and their own strategy $\alpha^i$. The vector $R_0 := (R_0^p, R_0^q)$ represents the reservation utilities (outside options) of the two market makers. We denote by $\mathcal{C}^{R_0}_i$ the projection of $\mathcal{C}^{R_0}$ onto the $i$-th component, representing the set of admissible rebate contracts $\xi^i$ available to Auction Maker $i \in \{p,q\}$.



\subsection{Auction makers' Cournot-Nash equilibrium}
\label{sec:agentprb}

Strategic traders aim to maximize their individual expected gains simultaneously and without communication, by dynamically controlling their arrival intensities and quoted prices. An auction maker's objective consists of three components:
\begin{itemize}
\item a terminal payoff $g^i(X_T)$ for $i \in \{p,q\}$, resulting from their executed trades in the auction;
\item a terminal incentive (rebate) $\xi^i$ received from the exchange;
\item a running penalty term penalizing deviations of their trading intensity from the target baseline $\lambda_0$. As discussed, this models a preference for tracking a benchmark TWAP execution strategy, penalized quadratically.
\end{itemize}

For each order submitted at price $P^i_t$, the executed volume at time $T$ is $Q^i(P^{cl}_T) = K^i(P^i_t - P^{cl}_T)$. If $P^{cl}_T \leq P^i_t$, the order is a buy order ($Q^i \geq 0$). The auction maker buys at the clearing price $P^{cl}_T$ and holds an asset valued at the efficient price $P^*_T$, yielding a profit per share of $(P^*_T - P^{cl}_T)$. The total profit for this order is therefore $K^i(P^i_t - P^{cl}_T)(P^*_T - P^{cl}_T)$. Conversely, if $P^{cl}_T > P^i_t$, the order is a sell order ($Q^i < 0$). The auction maker sells at $P^{cl}_T$ and buys back at $P^*_T$, yielding a profit per share of $(P^{cl}_T - P^*_T)$. In this case, the total profit is $K^i(P^i_t - P^{cl}_T)(P^*_T - P^{cl}_T)$, which is algebraically identical.

Therefore, the terminal payoff function $g^i(X_T)$ representing auction maker $i$'s total profit from trading at the end of the auction is given by:
\begin{align*}
    g^i(X_T) &:= \int_0^T K^i(P^{cl}_T - P^i_{s})(P^{cl}_T - P^{*}_T)dN^i_s \\
    &= K^i(P^{cl}_T - P^{*}_T)\left(P^{cl}_T N^i_T - \int_0^T P^i_{s}dN^i_s\right).
\end{align*}

Recalling the definitions of the state variables $X_t$, we can explicitly write the payoffs for both auction makers:
\begin{align*}
    g^p(X_T) &= K^p \left(\frac{X_T^1 - X_T^2 + K^pX_T^4 + K^{q}X_T^6}{K^pX_T^5+K^{q}X_T^7} - X_T^3\right)\left(X_T^5\frac{X_T^1 - X_T^2 + K^pX_T^4 + K^{q}X_T^6}{K^pX_T^5+K^{q}X_T^7} - X_T^4\right),
\end{align*}
and
\begin{align*}
    g^q(X_T) &= K^q \left(\frac{X_T^1 - X_T^2 + K^pX_T^4 + K^{q}X_T^6}{K^pX_T^5+K^{q}X_T^7} - X_T^3\right)  \left(X_T^7\frac{X_T^1 - X_T^2 + K^pX_T^4 + K^{q}X_T^6}{K^pX_T^5+K^{q}X_T^7} - X_T^6\right).
\end{align*}

Assume that Auction Maker $i \in \{p,q\}$ chooses the control $\alpha^i=(\lambda^i,\mu^i)$ while their opponent Auction Maker $j \neq i$ chooses $\alpha^j=(\lambda^j,\mu^j)$. Given a fixed transaction fee $d$ and a rebate scheme $\xi=(\xi^p,\xi^q)$ proposed by the exchange, the expected value function of Auction Maker $i$ at time $t=0$ is defined as:
\[ V^i_0(\xi^i,d,\alpha^j; \alpha^i) := \mathbb{E}^{\alpha} \left[ g^i(X_T) + \xi^i - \int_0^T \frac{(\lambda^i_s - \lambda_0 )^2}{2} ds \right]. \]

\begin{remark}[Indirect impact of transaction fees]
Note that the transaction fee $d$ is only paid directly by the non-strategic investors. However, it explicitly dictates the intensity of the processes $N^a$ and $N^b$, thereby affecting the terminal components $X^1_T$ and $X^2_T$ and the final clearing price $P^{cl}_T$. Consequently, the exchange's choice of $d$ indirectly but significantly impacts the value functions of the auction makers.
\end{remark}

Given a fixed rebate contract $\xi = (\xi^p,\xi^q)$ and a transaction fee $d$, the auction makers face a non-zero-sum stochastic differential game. The problem is to find an admissible strategy profile $\alpha^\star=(\alpha^{p,\star},\alpha^{q,\star}) \in \mathcal{A}$ that forms a Nash Equilibrium, satisfying the system:
\begin{equation}\label{eq:NE}
\text{NE}(\xi,d):
\begin{cases}
V^{p}_0(\xi^p,d,\alpha^{q,\star};\alpha^{p,\star}) &= \sup_{\alpha^p \in \mathcal{A}^p} V^p_0(\xi^p,d,\alpha^{q,\star}; \alpha^p) \\
V^{q}_0(\xi^q,d,\alpha^{p,\star};\alpha^{q,\star}) &= \sup_{\alpha^q \in \mathcal{A}^q} V^q_0(\xi^q,d,\alpha^{p,\star}; \alpha^q).
\end{cases}
\end{equation}
Such a profile $\alpha^\star$ satisfying \eqref{eq:NE} defines a Nash Equilibrium composed of optimal spreads $\mu^{p,\star}, \mu^{q,\star}$ and trading intensities $\lambda^{p,\star}, \lambda^{q,\star}$ under the given exchange policy.

\subsubsection{Best response of a representative Auction Maker} 

To characterize the Nash equilibrium $\text{NE}(\xi,d)$ under a specific rebate and transaction fee scheme, we first study the best-response behavior of a single strategic trader. Suppose Auction Maker $p$ optimizes their strategy while the control of their competitor, $\alpha^q$, is fixed and known. Auction Maker $p$ aims at finding an optimal control $\alpha^{p,\star}$ such that:
\begin{align}\label{eq:traderproblem}
    V^{p}_0(\xi^p,d,\alpha^{q};\alpha^{p,\star}) &= \sup_{\alpha^p \in \mathcal{A}^p} V^p_0(\xi^p,d,\alpha^{q}; \alpha^p) \nonumber \\
    &= \sup_{\alpha^p \text{ s.t.} ( \alpha^p,\alpha^q)\in \mathcal{A}} \mathbb{E}^{\alpha} \left[g^p(X_T) + \xi^p - \int_0^T \frac{(\lambda^p_s - \lambda_0 )^2}{2} ds \right].
\end{align}

To simplify the notation for the jumps of the 7-dimensional state process $X_t$, we introduce the canonical basis vectors of $\mathbb{R}^7$, denoted by $(\mathbf{e}_1, \dots, \mathbf{e}_7)$. We consider the following set of incentive-compatible contracts for Auction Maker $p$, given the fixed strategy $\alpha^q$:
\[ \Xi^p := \left\{ Y_T^{p, Y^p_0, Z^p, d} \; \big| \; Z^p \in \mathcal{Z}^p, \; Y_0^p \geq R_0^p \right\}, \]
where the dynamic value process $Y_t^{p, Y^p_0, Z^p, d}$ is defined as:
\begin{align*}
    Y^{p,Y^p_0, Z^p,d}_t &= Y^p_0 - \int_0^t F^p(X_s, Z_s^p, r_s, P_s, \alpha^q_s, d; \hat{\alpha}^p_s) ds + \int_0^t Z^{1,p}_s (dN^{a}_s - \lambda_0 ds) \\
    &\quad + \int_0^t Z^{2,p}_s (dN^{b}_s - \lambda_0 ds) + \int_0^t Z^{3,p}_s dW_s + \int_0^t Z^{4,p}_s (dN^p_s - \lambda_0 ds) \\
    &\quad + \int_0^t Z^{5,p}_s dW^{p}_s + \int_0^t Z^{6,p}_s (dN^{q}_s - \lambda_0 ds) + \int_0^t Z^{7,p}_s dW^{q}_s,
\end{align*} 

where the space $\mathcal{Z}^p$ denotes the set of $\mathbb{F}$-predictable processes $Z^p$ taking values in $\mathbb{R}^{7}$ such that:
\[ \sup_{d \in \mathbb{R}_+} \sup_{\alpha^p \in \mathcal{A}^p} \sup_{t \in [0,T]} \mathbb{E}^{\alpha} \left[ |Y_t^{p;Y_0^p,Z^p,d}| \right] < \infty \quad \text{and} \quad \mathbb E\left[\int_0^T |Z^p_s|^2 ds\right] < \infty.\]
Here, the random jump sizes are defined by $r_s := \left( v_{(N_{s^-}^a+1)}^a \mathbf{1}_{\{A_{(N_{s^-}^a+1)} \leq \theta(T - s)\}}, \; v_{(N_{s^-}^b+1)}^b \mathbf{1}_{\{B_{(N_{s^-}^b+1)} \leq \theta(T - s)\}} \right)$. 
$F^p$ is given by:
\begin{align*} 
    &F^p(x, z, r, P, \mu^q, \lambda^q, d; \mu^p, \lambda^p) \\
    &= \big( g^p(x + r^1\mathbf{e}_1) - g^p(x) \big) \lambda^a(\mu,d) + \big( g^p(x + r^2\mathbf{e}_2) - g^p(x) \big) \lambda^b(\mu,d) \\
    &\quad + \big( g^p(x + P^p\mathbf{e}_4 + \mathbf{e}_5) - g^p(x) \big) \lambda^p + \big(g^p(x + P^q\mathbf{e}_6 + \mathbf{e}_7) - g^p(x)\big) \lambda^{q} \\
    &\quad - \frac{(\lambda^p - \lambda_0 )^2}{2} - z^{1,p} (\lambda_0 - \lambda^a(\mu,d)) - z^{2,p} (\lambda_0 - \lambda^b(\mu,d)) \\
    &\quad - z^{4,p}(\lambda_0 - \lambda^p) - z^{6,p} (\lambda_0 - \lambda^{q}) - \frac{\mu^p}{\sigma} z^{5,p} - \frac{\mu^q}{\sigma} z^{7,p},
\end{align*}
and the optimal response $\hat{\alpha}^p_s := (\hat{\mu}^p_s, \hat{\lambda}^p_s)$ is defined as the pointwise maximizer:
\[ \hat{\alpha}^p(x,z,r,P,d,\alpha^q) \in \underset{\mu^p \in B, \; \lambda^p \in \Lambda}{\text{arg sup}} F^p(x,z,r,P,\alpha^q,d;\mu^p,\lambda^p). \]

\begin{remark}
By symmetry, $\Xi^q$ denotes the set of contracts $Y_T^{q;Y_0^q,Z^q,d}$ indexed by $Z^q \in \mathcal{Z}^q$ with the corresponding generator $F^q$ when $\alpha^p$ is fixed. 
\end{remark}
We deduce the best optimal response of Auction Maker p regarding a rebates $\xi^p$ proposed by the exchange and a trading strategy $\alpha^q$ used by the opponent Auction maker q as follows. The proof is postponed to Appendix \ref{app:bestreact}.
\begin{lemma} \label{lemma:trader}
Fix $\alpha^q \in \mathcal{A}^q$ and $d \in \mathbb{R}_+$.

\begin{enumerate}
    \item $\mathcal{C}^{R_0^p}_p = \Xi^p$. In other words, for any incentive-compatible contract $\xi^p \in \mathcal{C}^{R_0^p}_p$, there exists a unique process $Z^p \in \mathcal{Z}^p$ such that $\xi^p = Y_T^{p,Y_0^p,Z^p,d}$.
    \item Conversely, for every $\xi^p = Y_T^{p,Y_0^p,Z^p,d}$, Auction Maker $p$'s optimal control $\hat{\alpha}^p = (\hat{\mu}, \hat{\lambda})$ is given explicitly by:
    \begin{align*}
        \hat{\mu}(x,z,r,\mu^{q},d) &= \text{root of } \Bigg\{ -\frac{z_5}{\sigma} + \lambda_0 \iota(d) \big(z_1 + g^p(x + r_1\mathbf{e}_1) - g^p(x) \big) \frac{\partial i^a(\mu,\mu^{q})}{\partial \mu^p} \\
        &\quad\quad + \lambda_0 \iota(d) \big(z_2 + g^p(x + r_2\mathbf{e}_2) - g^p(x) \big) \frac{\partial i^b(\mu,\mu^{q})}{\partial \mu^p} \Bigg\} \vee (-\mu_{\infty}) \wedge \mu_{\infty}; \\
        \hat{\lambda}(x,z,P^p) &= 0 \vee \Big\{ \lambda_0 + z_4 + g^p(x + P^p\mathbf{e}_4 + \mathbf{e}_5) - g^p(x) \Big\} \wedge \lambda_{\infty}.
    \end{align*}
\end{enumerate} 
\end{lemma}

\subsubsection{Nash Equilibrium between Auction Makers}
\label{sec:nash}

We define the subset $U \subset \mathbb{R}^{7 \times 2}$ as the set of contract parameters $z$ such that for any state $(x,r,P,d) \in \mathbb{R}^7 \times \mathbb{R}^2 \times \mathbb{R}_+$, there exists a joint best-response $\hat{a}(x,z,r,P,d) = (\hat{a}^p(x,z,r,P,d); \hat{a}^q(x,z,r,P,d)) := (\hat{\mu}^p, \hat{\lambda}^p; \hat{\mu}^q, \hat{\lambda}^q)$ defined as the fixed point of the following system:
\begin{equation*}
(FP)
\begin{cases}
\hat{a}^p(x,z,r,P,d) \in \underset{\mu^p \in \Delta, \; \lambda^p \in \Lambda}{\text{arg sup}} F^p(x,z,r,P, \hat{a}^q(x,z,r,P,d),d;\mu^p,\lambda^p), \\
\hat{a}^q(x,z,r,P,d) \in \underset{\mu^q \in \Delta, \; \lambda^q \in \Lambda}{\text{arg sup}} F^q(x,z,r,P, \hat{a}^p(x,z,r,P,d),d;\mu^q,\lambda^q).
\end{cases}
\end{equation*}

\begin{remark}\label{rem:Ucharacterization} 
In this remark, we provide explicit examples and sufficient conditions ensuring that $U$ is non-empty. For notational convenience, we define the directional jump differences:
\[ Dg^{i,1}(x) := g^i(x + r^1 \mathbf{e}_1) - g^i(x), \quad Dg^{i,2}(x) := g^i(x + r^2 \mathbf{e}_2) - g^i(x). \]

\begin{enumerate}
    \item \textit{Rebates indexed on investors' activity only.} 
    Take $z^{k,i} = 0$ for any $k \in \{3,\dots,7\}$ and $i \in \{p,q\}$. This implies that only $z^{1,i}$ and $z^{2,i}$ are non-zero, meaning the exchange redistributes wealth solely based on the volume of buy and sell orders sent by non-strategic investors. Then, the system $(FP)$ is equivalent to:
    \[ 
    \begin{cases}
    \big( Dg^{p,1}(x)+z^{1,p} \big) \frac{\partial{\lambda^a}}{\partial \mu^p} (\hat{\mu},d) + \big( Dg^{p,2}(x)+z^{2,p} \big) \frac{\partial{\lambda^b}}{\partial \mu^p}(\hat{\mu},d) = 0, \quad \text{or } \hat{\mu}^p = \pm \mu_{\infty} \\
    \hat{\lambda}^p = \big( \lambda_0 + g^p(x + P^p\mathbf{e}_4 + \mathbf{e}_5) - g^p(x) \big) \\
    \big( Dg^{q,1}(x)+z^{1,q} \big) \frac{\partial{\lambda^a}}{\partial \mu^q} (\hat{\mu},d) + \big( Dg^{q,2}(x)+z^{2,q} \big) \frac{\partial{\lambda^b}}{\partial \mu^q}(\hat{\mu},d) = 0, \quad \text{or } \hat{\mu}^q = \pm \mu_{\infty} \\
    \hat{\lambda}^q = \big( \lambda_0 + g^q(x + P^q\mathbf{e}_6 + \mathbf{e}_7) - g^q(x) \big).
    \end{cases}
    \]
    The solution for $\hat{\lambda}$ is explicit. It remains to choose $z^{1,i}, z^{2,i}$ and the intensity functions $\lambda^a, \lambda^b$ to ensure the system for $\hat{\mu}$ admits a solution. For example, if we assume the exponential forms:
    \[\lambda^a = \lambda_0 e^{-c\mu^p - \tilde{c}\mu^q} e^{-d} \quad \text{and} \quad \lambda^b = \lambda_0 e^{c\mu^p + \tilde{c} \mu^q}e^{-d},\] 
    we obtain:
    \[ 
    \begin{cases}
    e^{2c\hat{\mu}^p + 2\tilde{c}\hat{\mu}^q} = \frac{Dg^{p,1}(x)+z^{1,p}}{Dg^{p,2}(x)+z^{2,p}} = \frac{Dg^{q,1}(x)+z^{1,q}}{Dg^{q,2}(x)+z^{2,q}}, \quad \text{or } \hat{\mu}^p, \hat{\mu}^q \in \{-\mu_{\infty}, \mu_{\infty}\} \\
    \hat{\lambda}^p = \big( \lambda_0 + g^p(x + P^p\mathbf{e}_4 + \mathbf{e}_5) - g^p(x) \big) \\
    \hat{\lambda}^q = \big( \lambda_0 + g^q(x + P^q\mathbf{e}_6 + \mathbf{e}_7) - g^q(x) \big).
    \end{cases}
    \]
    Consequently, the set $U$ contains all matrices $z \in \mathbb{R}^{7 \times 2}$ such that $z^{k,i} = 0$ for $k \in \{3,\dots,7\}$ and:
    \[ 
    \begin{cases}
    \frac{Dg^{p,1}(x)+z^{1,p}}{Dg^{p,2}(x)+z^{2,p}} = \frac{Dg^{q,1}(x)+z^{1,q}}{Dg^{q,2}(x)+z^{2,q}} > 0 \\
    Dg^{p,1}(x)+z^{1,p} < 0 \\
    Dg^{q,1}(x)+z^{1,q} < 0.
    \end{cases}
    \]

    \item \textit{Symmetric Auction Makers.}
    Assume that $\Lambda = [0,\lambda_\infty]$, $R_0^p = R_0^q$, $K^p = K^q$, and that the investor arrival intensities are perfectly symmetric with respect to the auction makers' controls:
    \[\lambda^a(\mu^p,\mu^q,\lambda^p,\lambda^q) = \lambda^a(\mu^q,\mu^p,\lambda^q,\lambda^p), \quad \lambda^b(\mu^p,\mu^q,\lambda^p,\lambda^q) = \lambda^b(\mu^q,\mu^p,\lambda^q,\lambda^p).\] 
    Due to symmetry, we have $Dg^{p,1} = Dg^{q,1} =: Dg^1$ and $Dg^{p,2} = Dg^{q,2} =: Dg^2$. Thus, $(FP)$ simplifies to:
    \[ 
    \begin{cases}
    -\frac{z^{5,p}}{\sigma} + (Dg^1(x)+z^{1,p})\frac{\partial \lambda^a}{\partial \mu^p}(\hat{\mu},d) + (Dg^2(x)+z^{2,p})\frac{\partial \lambda^b}{\partial \mu^p}(\hat{\mu},d) = 0, \quad \text{or } \hat{\mu}^p = \pm \mu_{\infty} \\
    -\frac{z^{7,q}}{\sigma} + (Dg^1(x)+z^{1,q})\frac{\partial \lambda^a}{\partial \mu^q}(\hat{\mu},d) + (Dg^2(x)+z^{2,q})\frac{\partial \lambda^b}{\partial \mu^q}(\hat{\mu},d) = 0, \quad \text{or } \hat{\mu}^q = \pm \mu_{\infty} \\
    \hat{\lambda}^p = \big( \lambda_0 + g^p(x + P^p\mathbf{e}_4 + \mathbf{e}_5) - g^p(x) + z^{4,p} \big) \\
    \hat{\lambda}^q = \big( \lambda_0 + g^q(x + P^q\mathbf{e}_6 + \mathbf{e}_7) - g^q(x) + z^{6,q} \big).
    \end{cases}
    \]
    In particular, if we use the symmetric exponential form $\lambda^a = \lambda_0 e^{-c(\mu^p+\mu^q)} e^{-d}$ and $\lambda^b = \lambda_0 e^{c(\mu^p + \mu^q)}e^{-d}$, we obtain:
    \[ 
    \begin{cases}
    -\frac{z^{5,p}}{\sigma} - c(Dg^1(x)+z^{1,p})e^{-c(\hat{\mu}^p + \hat{\mu}^q)} + c(Dg^2(x)+z^{2,p})e^{c(\hat{\mu}^p + \hat{\mu}^q)} = 0, \quad \text{or } \hat{\mu}^p = \pm \mu_{\infty} \\
    -\frac{z^{7,q}}{\sigma} - c(Dg^1(x)+z^{1,q})e^{-c(\hat{\mu}^p + \hat{\mu}^q)} + c(Dg^2(x)+z^{2,q})e^{c(\hat{\mu}^p + \hat{\mu}^q)} = 0, \quad \text{or } \hat{\mu}^q = \pm \mu_{\infty}.
    \end{cases}
    \]
    If we choose symmetric contract parameters $z^{1,p} = z^{1,q} = z^1$, $z^{2,p} = z^{2,q} = z^2$, and $z^{5,p} = z^{7,q} = \tilde{z}$, and if there exists a positive constant $\gamma$ satisfying both
    \begin{equation}\label{gamma}
    -\frac{\tilde{z}}{\sigma} - c(Dg^1(x)+z^{1})\frac{1}{\gamma} + c(Dg^2(x)+z^{2})\gamma = 0,
    \end{equation}
    and the second-order maximum condition
    \begin{equation}\label{eq:z_condition} 
    Dg^1(x)+z^{1} < 0, \quad Dg^2(x)+z^{2} < 0,
    \end{equation}
    then $(FP)$ admits a solution in $U$ given explicitly by:
    \[ 
    \begin{cases}
    \hat{\mu}^p + \hat{\mu}^q = \frac{\ln(\gamma)}{c} \\
    \hat{\lambda}^p = \big( \lambda_0 + g^p(x + P^p\mathbf{e}_4 + \mathbf{e}_5) - g^p(x) + z^{4,p} \big) \\
    \hat{\lambda}^q = \big( \lambda_0 + g^q(x + P^q\mathbf{e}_6 + \mathbf{e}_7) - g^q(x) + z^{6,q} \big).
    \end{cases}
    \]
\end{enumerate}
Therefore, the admissible set $U$ naturally contains the locus of points $z \in \mathbb{R}^{7 \times 2}$ such that $z^{1,p} = z^{1,q}$, $z^{2,p} = z^{2,q}$, $z^{5,p} = z^{7,q}$, and for which Equations \eqref{gamma} and \eqref{eq:z_condition} admit a real root $\gamma > 0$. 
\end{remark}

Our ultimate goal is to characterize the set of rebate contracts $(\xi,d)$ such that $\text{NE}(\xi,d)$ admits a solution. Extending the definition of $\Xi^p$, we now consider the following joint set of contracts:
\[ \Xi := \left\{ Y_T^{Y_0, Z, d} \;\big|\; Z \in \mathcal{Z}, \; Y_0^i \geq R_0^i, \; i \in \{p,q\} \right\}, \]
where the 2-dimensional dynamic value process $Y_t^{Y_0, Z, d} := (Y_t^{p; Z,d}, Y_t^{q; Z,d})$ is defined by the coupled It\^o's decomposition:
\begin{align*}
    Y^{i;Z,d}_t &= Y^i_0 - \int_0^t F^i(X_s, Z_s^i, r_s, P_s, \hat{\alpha}^j_s, d; \hat{\alpha}^i_s) ds + \int_0^t Z^{1,i}_s (dN^{a}_s - \lambda_0 ds) \\
    &\quad + \int_0^t Z^{2,i}_s (dN^{b}_s - \lambda_0 ds) + \int_0^t Z^{3,i}_s dW_s + \int_0^t Z^{4,i}_s (dN^p_s - \lambda_0 ds) \\
    &\quad + \int_0^t Z^{5,i}_s dW^{p}_s + \int_0^t Z^{6,i}_s (dN^{q}_s - \lambda_0 ds) + \int_0^t Z^{7,i}_s dW^{q}_s.
\end{align*} 
Here, $\hat{\alpha}^i_s := \hat{a}^i(X_s, Z_s, r_s, P_s, d)$, and $\mathcal{Z}$ denotes the space of random $\mathbb{F}$-predictable processes $Z = (Z^p, Z^q) \in \mathcal{Z}^p \times \mathcal{Z}^q$. Finally, we denote by $\mathcal{U}$ the subset of $\mathcal{Z}$ such that the process $Z$ is strictly $U$-valued.

Recall that an admissible pair of controls $\alpha = (\hat{\alpha}^p,\hat{\alpha}^q) \in \mathcal{A}$ constitutes a Nash equilibrium for the two auction makers if it solves the system:
\begin{equation*}
\text{NE}(\xi,d):
\begin{cases}
V^{p}_0(\xi^p,d,\hat{\alpha}^{q};\hat{\alpha}^{p}) &= \sup_{\alpha^p \in \mathcal{A}^p} V^p_0(\xi^p,d,\hat{\alpha}^{q}; \alpha^p) \\
V^{q}_0(\xi^q,d,\hat{\alpha}^{p};\hat{\alpha}^{q}) &= \sup_{\alpha^q \in \mathcal{A}^q} V^q_0(\xi^q,d,\hat{\alpha}^{p}; \alpha^q).
\end{cases}
\end{equation*}

Lemma \ref{lemma:trader} demonstrates that the optimal control for each auction maker solving $\text{NE}(\xi,d)$ depends not only on the exchange's incentives and transaction fees, but also heavily on the strategy of their competitor. Consequently, we aim to identify the specific set of incentives and fees that successfully induce a Nash equilibrium between the two auction makers. Without a properly designed incentive and fee structure, the auction makers might fail to reach an equilibrium, resulting in an endless cycle of strategy adjustments. It requires the formal characterization of the set of contracts $(\xi,d)$ such that $\text{NE}(\xi,d) \neq \emptyset$.

\begin{theorem}
\label{thm:main}
For a given contract $\xi \in \Xi$ and a transaction fee $d \in \mathbb{R}_+$, the following conditions are equivalent:

\begin{itemize}
    \item[(i)] $\text{NE}(\xi,d) \neq \emptyset$, meaning that there exists a Nash equilibrium $\alpha = ( \alpha^p,  \alpha^q) \in \mathcal{A}$;
    \item[(ii)] there exist unique initial values $Y_0 = (Y_0^p, Y_0^q) \in \mathbb{R}^2$ and a unique predictable process $Z = (Z^{p},Z^{q}) \in \mathcal{U}$ such that $\xi = (\xi^p, \xi^q) = (Y_T^{p, Y_0^p, Z^p, d}, Y_T^{q, Y_0^q, Z^q, d})$, where for each $i \in \{p,q\}$, the dynamic value process $Y^{i, Z^i, d}$ satisfies:
    \begin{align*}
        Y_t^{i,Z^i,d} &= Y^i_0 + \int_0^t Z^{1,i}_s (dN^{a}_s - \lambda_0 ds) + \int_0^t Z^{2,i}_s (dN^{b}_s - \lambda_0 ds) + \int_0^t Z^{3,i}_s dW_s \\
        &\quad + \int_0^t Z^{4,i}_s (dN^p_s - \lambda_0 ds) + \int_0^t Z^{5,i}_s dW^{p}_s + \int_0^t Z^{6,i}_s (dN^{q}_s - \lambda_0 ds) \\
        &\quad + \int_0^t Z^{7,i}_s dW^{q}_s - \int_0^t F^i(X_s,Z_s,r_s,P_s, \hat{\alpha}^j_s,d;\hat{\alpha}^i_s) ds,
    \end{align*}
    with $\hat{\alpha}^i_s = \hat{a}^i(X_s,Z_s,r_s,P_s,d)$ being the solution to the fixed-point system $(FP)$.
\end{itemize}

Furthermore, if the above conditions are satisfied, then the equilibrium strategy profile $\alpha$ from (i) coincides with the fixed point $\hat{\alpha}$ characterized in (ii).
\end{theorem}

\begin{proof}
    Assume that there exists a Nash equilibrium $\alpha=(\alpha^p,\alpha^q)$. Then, condition (ii) is a direct consequence of applying Lemma \ref{lemma:trader} simultaneously to both Auction Maker $p$ and Auction Maker $q$, under the assumption that there exists a solution to the fixed-point system $(FP)$. 
    
    Conversely, assume that condition (ii) is satisfied with the fixed-point control $\hat{\alpha}$. By Part 2 of Lemma \ref{lemma:trader}, $\hat{\alpha}^i$ is exactly the optimal best response of Auction Maker $i$ given the control $\hat{\alpha}^j$ played by the competing auction maker. Therefore, the pair of controls are mutually optimal, yielding $\hat{\alpha} \in \text{NE}(\xi,d)$, where the contract is generated as $\xi = (Y_T^{p, Y_0^p, Z^p, d}, Y_T^{q, Y_0^q, Z^q, d}) \in \Xi$.
\end{proof}

\begin{corollary}[Symmetric Auction Makers and Nash equilibrium] \label{cor:Nash}
Assume the symmetric setup defined in Remark \ref{rem:Ucharacterization}, where $\Lambda = [0,\lambda_\infty]$, $R_0^p = R_0^q$, $K^p = K^q$, and the intensities are given by:
\[\lambda^a(\mu,d) = \lambda_0 e^{-c(\mu^p+\mu^q)} e^{-d} \quad \text{and} \quad \lambda^b(\mu,d) = \lambda_0 e^{c(\mu^p + \mu^q)}e^{-d}.\]
Then, there exists a Nash equilibrium, i.e., $\text{NE}(\xi,d) \neq \emptyset$, if the contract is of the form $\xi = Y_T^{Y_0,Z,d} \in \Xi$, where the contract parameters satisfy $Z^{1,p}=Z^{1,q}$, $Z^{2,p}=Z^{2,q}$, $Z^{5,p}=Z^{7,q}$, and there exists a $\gamma > 0$ satisfying Equation \eqref{gamma}.

In this case, the optimal controls at equilibrium are explicitly given by:
\[ 
\begin{cases}
    \hat{\mu}^p_t + \hat{\mu}^q_t = \frac{\ln(\gamma)}{c} \\
    \hat{\lambda}^p_t = 0 \vee \big\{ \lambda_0 + g^p(X_{t^-} + P^p_t\mathbf{e}_4 + \mathbf{e}_5) - g^p(X_{t^-}) + Z^{4,p}_t \big\} \wedge \lambda_{\infty} \\
    \hat{\lambda}_t^q = 0 \vee \big\{ \lambda_0 + g^q(X_{t^-} + P^q_t\mathbf{e}_6 + \mathbf{e}_7) - g^q(X_{t^-}) + Z^{6,q}_t \big\} \wedge \lambda_{\infty}.
\end{cases}
\]
\end{corollary}

\section{Exchange's Problem: Optimal rebates and transaction fees}
\label{sec:principalprb}

The principal—the stock exchange—aims to minimize market inefficiency (defined as the squared deviation of the clearing price from the efficient price), minimize the total incentives paid out to the strategic auction makers, and maximize the total transaction fees collected from non-strategic investors. Recalling the definitions of the state variables $X^1_T$ and $X^2_T$ for the final aggregated non-cancelled volumes, we define the exchange's expected cost functional as:
\[ \rho(\xi,d,\alpha) = \mathbb{E}^{\alpha}\Big[ |P^{cl}_T - P_T^{*}|^2 + \xi^p + \xi^q - d(X^1_T + X^2_T) \Big]. \]
The exchange seeks to find the optimal admissible contract pair $(\xi,d)$ that minimizes $\rho(\xi,d,\alpha)$. \vspace{0.5em}

To formalize the exchange's problem in the presence of potentially multiple equilibria, we define the set of Principal-preferred equilibria as:
\[ \text{NEI}(\xi,d) = \big\{ \alpha \in \text{NE}(\xi,d) : \rho(\xi,d,\alpha) \leq \rho(\xi,d,\beta), \text{ for any } \beta \in \text{NE}(\xi,d) \big\}. \]

\begin{remark}[Optimistic approach to multiple equilibria]
The reason we introduce the set $\text{NEI}(\xi,d)$ is because certain contract pairs $(\xi,d)$ might induce a game that admits multiple Nash equilibrium strategy profiles. In line with the standard mechanism design literature, we adopt the ``optimistic approach'' and assume that the auction makers will select the Nash equilibrium that yields the greatest benefit (lowest cost) to the exchange. In practice, the exchange effectively proposes recommended prices and arrival intensities to the strategic traders along with the contract, serving as a coordinating signal, as explained in \cite{elie2019contracting}. Unlike \cite[Assumption 3.1]{hernandez2024principal}, we do not assume the uniqueness of the Nash equilibrium a priori, but rather defer the selection to the exchange's recommended actions, following the seminal contracting framework of Holmström and Milgrom \cite{holmstrom1987aggregation}.
\end{remark}

For any fee $d \in \mathbb{R}_+$, we define $\Xi(d) = \{ \xi \in \Xi : \text{NE}(\xi,d) \neq \emptyset \}$ as the set of implementable rebate contracts that successfully induce a Nash equilibrium between the auction makers. 

Let $\hat{\alpha}(\xi,d) \in \text{NEI}(\xi,d)$ denote the equilibrium strategy. We define the exchange's  optimization problem as:
\begin{align}
    \rho &= \min_{d>0, \; \xi \in \Xi(d)} \rho(\xi, d, \hat{\alpha}(\xi,d)) \nonumber \\
    &= \min_{d > 0, \; \xi \in \Xi(d)} \mathbb{E}^{\hat{\alpha}(\xi,d)}\Big[ |P^{cl}_T - P_T^{*}|^2 + \xi^p + \xi^q - d(X^1_T + X^2_T) \Big] \label{eq:ex} \\
    &\text{subject to } \quad V_0^{p}(\xi^p,d,\hat{\alpha}) \geq R_0^p, \quad V_0^{q}(\xi^q,d,\hat{\alpha}) \geq R_0^q, \label{eq:ex_const}
\end{align}
where $R_0 = (R_0^p, R_0^q) \in \mathbb{R}^2$ is the vector of reservation utilities for Auction Makers $p$ and $q$. The values $R_0^p$ and $R_0^q$ can be naturally interpreted, for example, as the expected PnL each auction maker would guarantee if they simply provided liquidity on a co-existing continuous limit order book at the fair price $P^*$ over the duration $T$.

Due to the simplicity of the transaction fee structure for non-strategic investors (parameterized entirely by the scalar $d$), the problem can be decoupled. We rely on numerical methods to find the globally optimal fee $\hat{d}$. Fixing a candidate $d>0$, the exchange first optimizes the rebate contract dynamically by solving the inner continuous-time optimization problem:
\begin{equation}\label{eq:exchange_incentive}
    \rho(d) = \min_{\xi \in \Xi(d)} \mathbb{E}^{\hat{\alpha}(\xi,d)}\Big[ |P^{cl}_T - P_T^{*}|^2 + \xi^p + \xi^q - d(X^1_T + X^2_T) \Big],
\end{equation}
yielding the conditionally optimal rebate policy $\hat{\xi}(d)$. The global solution is then found by minimizing $\rho(d)$ over $d > 0$.

\begin{remark}
While the optimal rebate contract $\xi$ can be characterized analytically by verification (see the section below) for a fixed $d$, the global objective function $\rho(d)$ is fundamentally intractable for joint analytical optimization, necessitating a numerical approach for $\hat{d}$. The fee $d$ enters the objective in highly non-linear ways that may break global convexity. First, in the revenue term, $d$ linearly scales the fees but exponentially suppresses the non-strategic arrival intensities (e.g., $\lambda \propto e^{-d}$), creating a concave revenue profile for large $d$. More importantly, $d$ is deeply embedded inside the market inefficiency penalty $\mathbb{E}[|P^{cl}_T - P_T^{*}|^2]$. The clearing price $P^{cl}_T$ depends on the final state of the cumulative order flow, which is governed by Poisson processes whose intensities are functions of both $d$ and the strategic best-responses $\hat{\alpha}(\xi,d)$ of the auction makers. We therefore decouple the problem, optimizing $\xi$ theoretically for a fixed $d$ and relying on numerical methods to find the optimal scalar $\hat{d}$ with grid search.
\end{remark}
\subsection{Saturated utility reservations}\label{section:saturated}

Recalling the characterization in Theorem \ref{thm:main}, solving the exchange's problem is equivalent to finding optimal controls $(\hat{Y}_0^p, \hat{Y}_0^q, \hat{Z}^p, \hat{Z}^q)$ such that $\xi^p \in \Xi^p$ and $\xi^q \in \Xi^q$. 

Because the auction makers are risk-neutral, the dynamic value process translates linearly with its initial value, meaning $Y_T^{i, Y_0^i, Z^i} = Y_0^i + Y_T^{i, 0, Z^i}$ for $i \in \{p,q\}$. This structural property allows us to separate the initial utility allocation from the dynamic stochastic incentives. Therefore, the exchange's objective function can be rewritten as:
\begin{align*}
    \rho(d) &= \inf_{\substack{Y^p_0 \geq R_0^p, \; Y^q_0 \geq R_0^q \\ Z \in \mathcal{Z}}} \mathbb{E}^{\hat{\alpha}(Z,d)}\Big[ |P^{cl}_T - P_T^{*}|^2 + Y_T^{p, Y_0^p, Z^p} + Y_T^{q, Y_0^q, Z^q} - d(X_T^1 + X_T^2) \Big] \\
    &= \inf_{\substack{Y^p_0 \geq R_0^p, \; Y^q_0 \geq R_0^q \\ Z \in \mathcal{Z}}} \mathbb{E}^{\hat{\alpha}(Z,d)}\Big[ |P^{cl}_T - P_T^{*}|^2 + Y_0^p + Y_0^q + Y_T^{p, 0, Z^p} + Y_T^{q, 0, Z^q} - d(X_T^1 + X_T^2) \Big] \\
    &= \inf_{Y^p_0 \geq R_0^p, \; Y^q_0 \geq R_0^q} \big( Y_0^p + Y_0^q \big) + \inf_{Z \in \mathcal{Z}} \mathbb{E}^{\hat{\alpha}(Z,d)}\Big[ |P^{cl}_T - P_T^{*}|^2 + Y_T^{p, 0, Z^p} + Y_T^{q, 0, Z^q} - d(X_T^1 + X_T^2) \Big] \\
    &= R_0^p + R_0^q + \inf_{Z \in \mathcal{U}} \mathbb{E}^{\hat{\alpha}(Z,d)}\Big[ |P^{cl}_T - P_T^{*}|^2 + Y_T^{p, 0, Z^p} + Y_T^{q, 0, Z^q} - d(X_T^1 + X_T^2) \Big].
\end{align*}
Consequently, the optimization problem for the exchange reduces entirely to the continuous-time optimization of the $Z$ process:
\begin{equation}\label{opt:tilderho} 
  V^P_0(x,0,r,d):=  \inf_{Z \in \mathcal{Z}} \mathbb{E}^{\hat{\alpha}(Z,d)}\Big[ |P^{cl}_T - P_T^{*}|^2 + Y_T^{p, 0, Z^p} + Y_T^{q, 0, Z^q} - d(X_T^1 + X_T^2) \Big].
\end{equation}

\subsection{Verification theorems}
We turn to the solution of \eqref{opt:tilderho} by using a verification argument. We first define $\mathcal C'([0,T],\mathcal S)$ as the set of function $\phi$ defined from $t\in [0,T] \times (x,y,r,p;z,d)\in \mathcal S= \mathbb R^7\times \mathbb R^2\times \mathbb R^2\times \mathbb R^2\times \mathbb R_+$ into $\mathbb R$, which are continuously differentiable in time $t$, twice continuously differentiable with respect to the variables $x_3,p,y$ and continuous in the other variables. We then define the Dynkin operator $\mathcal L^z$ associated with the processes $X,Y$ for any $\phi\in \mathcal C'([0,T],\mathcal S) $ and $z\in \mathbb R^{7\times 2}$ by 
\begin{align*}
    &\mathcal L^z\phi(t,x,y,r,p;d)\\
    &:= \sum_{i = p,q} \frac{\partial \phi}{\partial y_i} (t, x, y,r,p) \Bigl\{\lambda_0(z^{1,i} + z^{2,i} +z^{4,i} +z^{6,i} ) - F^i(x,z,r,p, \hat \alpha^j_s,d;\hat \alpha^i_s) -  z^{5,i} \frac{\hat\mu^p}{\sigma} - z^{7,i} \frac{\hat\mu^q_s}{\sigma} \Bigl\} \\
    &+\frac{1}{2} \frac{\partial^2 \phi}{\partial x_3 \partial x_3} (t, x,y,r,p) \sigma^2 + \frac{1}{2}\sum_{i,j=p,q} \frac{\partial^2 \phi}{\partial y_i \partial y_j} (t, x,y,r,p) \Bigl\{z^{3,i} z^{3,j} + z^{5,i}z^{5,j}+z^{7,i}z^{7,j}\Bigl\}\\
    &- \sum_{j=p,q}  \frac{\partial^2 \phi}{\partial x_3 \partial y_j} (t, x,y,r,p) \sigma z^{3,j}  + \sum_{i=p,q} \frac{\partial \phi}{\partial p_i} (t, x,y,r,p)\hat\mu_t^i +\frac{1}{2}\sum_{i =p,q}  \frac{\partial^2 \phi}{\partial p_i \partial p_i} (t, x,y,r,p) \sigma^2 \\
    &-\sum_{i = p,q}\frac{\partial^2 \phi}{\partial y_i \partial p_p} (t, x,y,r,p) \sigma z^{5,i} - \sum_{i = p,q}\frac{\partial^2 \phi}{\partial y_i \partial p_q} (t, x,y,r,p) \sigma z^{7,i} \\
    &+  [\phi(t,x+[r_1 , 0,0, 0,0,0,0],y +[-z^{1,p},-z^{1,q}],r,p)  - \phi(t, x,y,r,p)]  \lambda^a_t \\
    &+ [\phi(t,x+[0 , r_2,0, 0,0,0,0],y +[-z^{2,p},-z^{2,q}],r,p)  - \phi(t, x,y,r,p)]  \lambda^b_t\\
    &+ [\phi(t,x+[0 , 0,0, p^p,1,0,0],y +[-z^{4,p},-z^{4,q}],r,p)  - \phi(t, x,y,r,p)] \hat\lambda^p_t\\
    &+ [\phi(t,x+[0 , 0,0, 0,0,p^q,1],y +[-z^{6,p},-z^{6,q}],r,p)  - \phi(t, x,y,r,p)] \hat\lambda^q_t,
\end{align*}
where  $\lambda^a,\lambda^b,\hat\alpha=(\hat\alpha^p;\hat\alpha^q)=(\hat\mu^p,\hat\lambda^p;\hat\mu^q,\hat\lambda^p)$, are functions of $(x,z,p,d)$ from Theorem \ref{thm:main}.

\begin{theorem}[Verification and optimal rebates]\label{thm:hjb}Let $d>0$ be fixed. Assume that
        \begin{itemize}
            \item there exists a function $\phi \in C'([0,T],\mathcal{S})  $ satisfying\\
            \begin{equation}
            \left\{
            \begin{aligned}
                &\frac{\partial \phi}{\partial t} + \inf_{z\in \mathbb R^{7\times 2}} \mathcal L^z\phi(t,x,y,r,p;d) = 0,\text{ }t<T, \\
                &\phi(T,x,y,r,p;d) = \left(\frac{x_1 - x_2 + K^p x_4 + K^q x_6 + K^0 P_0^*}{K^p x_5 + K^q x_7 + K^0} - x_3\right)^2 + y \cdot I_2 - d(x_1 + x_2),\; (x,y,r,p)\in \mathcal S;
            \end{aligned}
            \right.
            \label{eq:zsolution} 
            \end{equation}
            
            \item there exists $n_1,n_2, n_3 \in \mathbb{N}$ and $C \in \mathbb{R}$ s.t. $|\phi(t,x,y,r,p,d)| < C(1+ |x|^{2}+|y|+|r|^{n_1}+|p|^{n_2})$ for every  $(t,x,y,r,p) \in [0,T] \times\mathcal{S}$;
            \item for every $(t,x,y,r,p) \in [0,T] \times \mathcal{S} $, there exists $\hat z(d)=\hat{z}(t,x,y,r,p;d) \in \mathbb{R}^{7\times 2}$ such that 
            \[
            \mathcal L^{\hat z(d)}\phi(t,x,y,r,p;d) =  \inf_{z\in \mathbb R^{7\times 2}}\mathcal L^z\phi(t,x,y,r,p;d)\] and such that $\hat Z_t(d)=\hat{Z}(t,X_t,Y_t,r_t,P_t;d) \in \mathcal{Z}$.
        \end{itemize}
        Then $\rho(d) =R_0^p  +R_0^q+ \phi(0,X_0, Y_0, r_0,p_0)$ and $\hat Z(d)$ is the optimal control of the exchange's problem \eqref{opt:tilderho}. Consequently, the optimal rebate policy is 
        \[\xi^p=Y
        _T^{p,\hat Z^p(d),d},\; \xi^p=Y
        _T^{q,\hat Z^q(d),d}, \]
        where $Y _T^{i,z,d}$ is defined by Theorem \ref{thm:main}.
\end{theorem}
\begin{proof}
\label{proof:verification}
   The proof is based on a classical verification theorem based on a localisation procedure and by applying Ito's formula. Define a local time $\theta_n =T \wedge \inf\{ t>0:||X_t|| >n, ||Y_t|| >n ,|r_t| >n, |P_t| >n \}$. By applying Ito's formula between $0$ and $\theta_n$, we get for any control $Z\in \mathcal Z$ \begin{align}
    \begin{split}
    &\phi(\theta_n,X_{\theta_n},Y^{Z,d}_{\theta_n}, r_{\theta_n},P_{\theta_n}) \\
    = &\phi(0,X_0,Y_0,r_0,P_0) + \int_0^{\theta_n}[ \frac{\partial \phi}{\partial s} (s, X_s, Y^{Z,d}_s,r_s,P_s) +\mathcal L\phi (s, X_s, Y^{Z,d}_s,r_s,P_s; Z_s, d)]ds\\
    &+ \int_0^{\theta_n} [\phi(s,X_{s-}+[r^1_s , 0,0, 0,0,0,0],Y^{Z,d}_{s^-} +[Z^{1,p}_s,Z^{1,q}_s],r_{s-},P_{s-})  - \phi(s,X_{s-},Y^{Z,d}_{s-},r_{s-},P_{s-})]  d\tilde N_s^a\\
    &+ \int_0^{\theta_n} [\phi(s,X_{s-}+[0 , r^2_s,0, 0,0,0,0],Y^{Z,d}_{s^-} +[Z^{2,p}_s,Z^{2,q}_s],r_{s-},P_{s-})  - \phi(s,X_{s-},Y_{s-},r_{s-},P_{s-})]  d\tilde  N_s^b\\
    &+ \int_0^{\theta_n} [\phi(s,X_{s-}+[0 , 0,0, P_s^p,1,0,0],Y^{Z,d}_{s^-} +[Z^{4,p}_s,Z^{4,q}_s],r_{s-},P_{s-})  - \phi(s,X_{s-},Y_{s-},r_{s-},P_{s-})]  d\tilde  N_s^p\\
    &+ \int_0^{\theta_n} [\phi(s,X_{s-}+[0 , 0,0, 0,0,P_s^q,1],Y^{Z,d}_{s^-} +[Z^{6,p}_s,Z^{6,q}_s],r_{s-},P_{s-})  - \phi(s,X_{s-},Y^{Z,d}_{s-},r_{s-},P_{s-})]  d\tilde  N_s^q,
    \end{split}
    \label{eq:phi_ito}
\end{align}
where $\tilde N^i_t:=N^i_t-\int_0^t\lambda^i_sds$ for any $i=a,b$ and $\tilde N^j_t:=N^j_t-\int_0^t\hat\lambda^j_tdt$ for $j=p,q$ are $\mathbb P^{\hat\alpha}-$martingales.

By the continuity of $\phi$ and its derivative and by the boundedness of all the state variables $X,Y,r,P$ on $[0,\theta_n]$, we deduce that 
\begin{align*}
&\mathbb{E}^{\hat\alpha} [\phi(\theta_n,X_{\theta_n},Y^{Z,d}_{\theta_n}, r_{\theta_n},P_{\theta_n})]\\
&= \mathbb{E}^{\hat\alpha} \Big[\phi(0,X_0,Y_0,r_0,P_0) + \int_0^{\theta_n}  \Bigr\{ \frac{\partial \phi}{\partial s} (s, X_s, Y^{Z,d}_s,r_s,P_s) ds  +\mathcal L\phi (s, X_s, Y^{Z,d}_s, r_s, P_s; Z_s) \Bigr\} ds \Big] \\
&\geq \phi(0,X_0,Y_0,r_0,P_0), \text{ for any } Z\in \mathcal Z,\end{align*} where equality is attained at $Z=\hat Z$ by the definition of $\phi$ solving \eqref{eq:zsolution}.

Note that
\begin{align*}
|\phi(\theta_n,X_{\theta_n},Y^{Z,d}_{\theta_n}, r_{\theta_n},P_{\theta_n})| &< C(1+ |X_{\theta_n}|^{2}+|Y_{\theta_n}|+|r_{\theta_n}|^{n_1}+|P_{\theta_n}|^{n_2})\\
&\leq \sup_{t \in [0,T]} C(1+ |X_t|^{2}+|Y^{Z,d}_t|+|r_t|^{n_1}+|P_t|^{n_2}).\end{align*} For $i = p,q$, we have \[ \sup_{t \in [0,T]} |P^i_t|^{n_2} \leq c
_{n_2} \Big((T\mu_\infty)^{n_2}+  \sup_{t\in [0,T]} |\sigma \tilde W^i_t|^{n_2}\Big) ,\] for some constant $c_{n_2}>0$ independent of $n$. Recall now from the proof of Lemma \ref{lemma:trader} that $\mathbb{E}^{\hat\alpha} \sup_{t \in [0,T]} |X_t|^{2} < \infty$. As $Z \in \mathcal{Z}$, we also have $\mathbb{E}^{\hat\alpha} \sup_{t \in [0,T]} |Y_t| < \infty$. We then deduce from the dominated convergence theorem that  \begin{align*}
&\mathbb{E}^{\hat\alpha} [|P^{cl}_T - P_T^{*}|^2 + Y_T \cdot I_2 - d(X_T^1 + X_T^2)]= \mathbb{E}^{\hat\alpha} \phi(T,X_T,Y_T, r_T,P_T;d)\\
&= \lim_{n \to \infty} \mathbb{E}^{\hat\alpha} \phi(\theta_n,X_{\theta_n},Y_{\theta_n}, r_{\theta_n},P_{\theta_n};d) \geq \phi(0,X_0,Y_0,r_0,P_0;d),\end{align*} for any $Z\in \mathcal Z$ with equality for $Z=\hat Z(d)$. Then $\rho(d) = \phi(0,X_0,Y_0,r_0,P_0;d)$ and $\xi=Y_T^{\hat Z(d),d}$ is the optimal rebate scheme.
\end{proof}
We now turn to a reduction of dimension technique due to the structure of our problem. Note that the dynamic version of the exchange value function is
\begin{equation*}
        V^P(t, x, y, p, r; d) := \inf_{Y^{p}_0 \ge R^p_a, \, Y^{q}_0 \ge R^q_a} \inf_{Z\in \mathcal Z} \mathbb{E}^{\alpha(Z,d)} \left[ \left| P^{cl}_T - P^*_T \right|^2 + Y_T^{p,0,Z^p}+Y_T^{q,0,Z^q} - d (X^1_T + X^2_T) \mid \mathcal{C} \right],
    \end{equation*}
    where $\mathcal{C} = \{X_t = x, Y_t = y, P^*_t = p, R_t = r, d\}$. We note that the payoff of the exchange is linear with respect to the $Y$ components, suggesting a reduction of dimension. We introduce the reduced infinitesimal generator $\tilde{\mathcal L}$ for any $\varphi\in \mathcal C'([0,T],\mathbb R^7\times \mathbb R^2\times \mathbb R^2\times  \mathbb R_+) $ and any $z\in \mathbb R^{7\times 2}$.
\begin{align*}
    &\tilde{\mathcal L}^z\varphi(t,x,r,p;d)\\
    &:= \sum_{i = p,q}  \Bigl\{\lambda_0(z^{1,i} + z^{2,i} +z^{4,i} +z^{6,i} ) - F^i(x,z,r,p, \hat \alpha^j_s,d;\hat \alpha^i_s) -  z^{5,i} \frac{\hat\mu^p_t}{\sigma} - z^{7,i} \frac{\hat\mu^q_t}{\sigma} \Bigl\} \\
    &+\frac{1}{2} \frac{\partial^2 \varphi}{\partial x_3 \partial x_3} (t, x,r,p) \sigma^2  + \sum_{i=p,q} \frac{\partial \varphi}{\partial p_i} (t, x,r,p)\hat\mu_t^i +\frac{1}{2}\sum_{i =p,q}  \frac{\partial^2 \phi}{\partial p_i \partial p_i} (t, x,r,p) \sigma^2 \\
    &+  [\varphi(t,x+[r_1 , 0,0, 0,0,0,0],r,p)  - \varphi(t, x,r,p) -z^{1,p}-z^{1,q}]  \lambda^a_t \\
    &+ [\varphi(t,x+[0 , r_2,0, 0,0,0,0],r,p)  - \varphi(t, x,r,p) -z^{2,p}-z^{2,q}]  \lambda^b_t\\
    &+ [\varphi(t,x+[0 , 0,0, p^p,1,0,0],r,p)  - \varphi(t, x,r,p) -z^{4,p}-z^{4,q}] \hat\lambda^p_t\\
    &+ [\varphi(t,x+[0 , 0,0, 0,0,p^q,1],r,p)  - \varphi(t, x,r,p) -z^{6,p}-z^{6,q}] \hat\lambda^q_t.
\end{align*}
We thus have the following result from direct computations. 
\begin{proposition}
    Let $\varphi$ be a $C'([0,T],\mathbb R^7\times  \mathbb R^2\times \mathbb R^2\times \mathbb R_+)$ solution to the following integro-partial HJB equation: 
       \begin{equation}
            \left\{
            \begin{aligned}
                &\frac{\partial \varphi}{\partial t} + \inf_{z\in \mathbb R^{7\times 2}} \tilde{\mathcal L}^z\varphi(t,x,r,p;d) = 0,\text{ }t<T, \\
                &\varphi(T,x,r,p;d) = \left(\frac{x_1 - x_2 + K^p x_4 + K^q x_6 + K^0 P_0^*}{K^p x_5 + K^q x_7 + K^0} - x_3\right)^2 - d(x_1 + x_2),\; (x,r,p)\in \mathcal S;
            \end{aligned}
            \right.
            \label{eq:zsolution:reduce} 
            \end{equation}
            Then, there exists a solution to \eqref{eq:zsolution} given by
            \[  \phi(t,x,y,r,p;d)=\varphi(t,x,r,p;d)+y\cdot I_2. \]
\end{proposition}

\subsection{Interpretations and key financial insights}
\label{sec:interpretresult}
We now turn to the financial interpretation of our results, with a focus on the structure of the optimal rebates and the corresponding optimal controls of the strategic traders. 

\paragraph{Optimal rebate policy.} 
Recall that the optimal rebate structure is given by:
\begin{align*}
    \xi^i =& R^i_0 + \int_0^T \hat{Z}^{1,i}_s (dN^{a}_s - \lambda_0 ds) + \int_0^T \hat{Z}^{2,i}_s (dN^{b}_s - \lambda_0 ds) + \int_0^T \hat{Z}^{3,i}_s dW_s \\
    &+ \int_0^T \hat{Z}^{4,i}_s (dN^p_s - \lambda_0 ds) + \int_0^T \hat{Z}^{5,i}_s dW^{p}_s + \int_0^T \hat{Z}^{6,i}_s (dN^{q}_s - \lambda_0 ds) \\
    &+ \int_0^T \hat{Z}^{7,i}_s dW^{q}_s - \int_0^T F^i(X_s,\hat{Z}_s,r_s,P_s,d) ds, \quad i \in \{p,q\},
\end{align*}
where, for simplicity of notation, we set $\hat{Z} = \hat{Z}(d)$. 

Each strategic trader receives an incentive indexed to the canonical process $X$, penalized by the continuous-time drift $F^i$ representing the certainty equivalent of their utility. For each strategic trader $i \in \{p,q\}$, the terms can be interpreted as follows:
\begin{itemize}
    \item $\hat{Z}^{1,i}$ and $\hat{Z}^{2,i}$ represent the incentives tied to the arrival of non-strategic buyers and sellers in the auction, respectively;
    \item $\hat{Z}^{3,i}$ denotes the incentive corresponding to fluctuations in the fundamental price $P^\star$;
    \item $\hat{Z}^{4,i}$ and $\hat{Z}^{6,i}$ are the incentives linked to the successful order placements by strategic traders $p$ and $q$, respectively;
    \item $\hat{Z}^{5,i}$ and $\hat{Z}^{7,i}$ are the incentives indexing the anchor prices of each strategic trader via the Brownian motions $W_s^p$ and $W_s^q$, respectively;
    \item the continuous time component $F^i$ is the inner gain of strategic trader $i$ resulting from their optimization (the certainty equivalent) that the exchange redistributes among the agents. 
\end{itemize}

The $F^i$ component is responsive to each strategic trader's gains and losses throughout the auction. Recall that a strategic trader's payoff from the market state at time $t$ is $g^i(X_t)$. When a new order enters the market at time $t$, the net gain for trader $i$ is $g^i(X_t) - g^i(X_{t^-})$. We say that the trader incurs a loss if this net gain is negative, and profits if it is positive. Observing the structure of $F^i$, the incentive $\xi^i$ has a positive increment when a strategic trader experiences a loss from a newly arrived order, and a negative increment if they profit from it. In short, the incentive $\xi^i$ compensates each strategic trader when adverse market movements harm them, and charges them when market changes benefit them. Additionally, the $F^i$ component depends on the fee $d$ charged to non-strategic traders. This indicates that the strategic traders are influenced not only by incentives directed explicitly at them, but also by the fee structure imposed on other market participants.

Given this rebate policy, we now interpret the strategic traders' equilibrium controls $\hat{\alpha} = (\hat{\mu}^p, \hat{\lambda}^p ; \hat{\mu}^q, \hat{\lambda}^q)$. 

\paragraph{Optimal trading speed in the auctions.} 
The optimal arrival intensity for strategic trader $p$ is $\hat{\lambda}^p(X_s,\hat{Z}_s,P_s^p) = 0 \vee \{\lambda_0 + \hat{Z}_s^{4,p} + g^p(X_s + P_s^p\mathbf{e}_4 + \mathbf{e}_5) - g^p(X_s)\} \wedge \lambda_{\infty}$, and for strategic trader $q$ is $\hat{\lambda}^q(X_s,\hat{Z}_s,P_s^q) = 0 \vee \{\lambda_0 + \hat{Z}_s^{6,q} + g^q(X_s + P_s^q\mathbf{e}_6 + \mathbf{e}_7) - g^q(X_s)\} \wedge \lambda_{\infty}$. This implies that the strategic traders anchor at the baseline intensity $\lambda_0$ and dynamically adjust their speeds. Traders submit orders more frequently if they stand to gain from their arrivals, i.e., when $g^p(X_s + P_s^p\mathbf{e}_4 + \mathbf{e}_5) - g^p(X_s) > 0$. Furthermore, when the exchange rewards a trader's arrival (i.e., when $\hat{Z}^{4,p}_s$ or $\hat{Z}^{6,q}_s$ are positive), the trader increases their intensity. Conversely, if the exchange penalizes arrivals, they slow down their trading activities. This mechanism allows the exchange to monitor and regulate trader activity.

\paragraph{Optimal spread.}
To illustrate the optimal spread, we consider the specific example where the non-strategic impact functions are $i^a(\mu^p, \mu^q) = e^{-c(\mu^p + \mu^q)}$ and $i^b(\mu^p, \mu^q) = e^{c(\mu^p + \mu^q)}$ for some constant $c>0$, assuming the symmetric strategic trader case investigated in Remark \ref{rem:Ucharacterization}. Let $z^{1,p} = z^{1,q} = z^1$, $z^{2,p} = z^{2,q} = z^2$, and $z^{5,p} = z^{7,q} = \tilde{z}$. Assume there exists a constant $\gamma > 0$ such that:
\begin{equation}\label{eq:gamma_nash}
    -\frac{\tilde{z}}{\sigma} - c\big(Dg^1(x) + z^1\big)\frac{1}{\gamma} + c\big(Dg^2(x) + z^2\big)\gamma = 0, \quad Dg^1(x) + z^1 < 0, \quad Dg^2(x) + z^2 < 0.
\end{equation}
Then, the optimal strategies satisfy:
\[ 
\begin{cases}
    \hat{\mu}^p + \hat{\mu}^q &= \frac{\ln(\gamma)}{c} \\[0.5em]
    \hat{\lambda}^p &= \big( \lambda_0 + g(x + P^p\mathbf{e}_4 + \mathbf{e}_5) - g(x) + z^{4,p} \big) \\[0.5em]
    \hat{\lambda}^q &= \big( \lambda_0 + g(x + P^q\mathbf{e}_6 + \mathbf{e}_7) - g(x) + z^{6,q} \big).
\end{cases}
\]
If such a $\gamma$ exists, any rebate contract $\xi = Y_T^{Y_0,Z}$ that induces a non-empty equilibrium set $\text{NEI}(\xi^{Y_0,Z},d)$ must satisfy this relation for $\{z^{1,p}, z^{2,p}, z^{5,p}, z^{1,q}, z^{2,q}, z^{7,q}\}$. Note that $\text{NEI}(\xi,d)$ will contain uncountably many pairs of $(\mu^p, \mu^q)$ such that $\mu^p + \mu^q = \ln(\gamma)/c$ for any valid $\gamma$. This demonstrates that it is the combined price impact of the two strategic traders that matters to the market and the exchange, rather than the individual proposed price of each trader.

\section{Numerical methods: Deep neural networks for high-dimensional integro-HJB equations}
\label{sec:neural}

Although Section \ref{sec:principalagent} implicitly identifies the optimal incentive structure $\xi$ for the strategic traders, it does not provide explicit solutions to the integro-partial differential equation considered in Theorem \ref{thm:hjb}. It remains to determine how an exchange can implement this incentive in practice and calculate the optimal fee $d$ for the non-strategic traders. To address this, we present a numerical method designed to explicitly compute the optimal incentive structure and the optimal fee $d$, leveraging the Deep BSDE method \cite{crandall1983viscosity,henry2017deep,han2020convergence,ji2020three,ji2022deep}. \vspace{0.5em}

Recall from Theorem \ref{thm:main} that finding the optimal incentive structure $\xi$ is equivalent to finding an optimal pair $(Y_0,Z)$. With the constraint \eqref{eq:ex_const} and Section \ref{section:saturated}, this simplifies to finding an optimal pair $(0,Z)$. Theorem \ref{thm:hjb} shows that the optimal $Z$ can be viewed as a function of $\{t,X_t, Y_t, r_t,P_t;d\}$. In order to validate our Deep BSDE methods used for the numeric, we first confirm that there exists a unique solution in a weak sense to the reduced PDE \eqref{eq:zsolution:reduce}. We use the notion of viscosity solution \cite{barles2008second,jing2013regularity}

\begin{definition}[Viscosity solution]
\label{def:viscosity}
Let $\mathcal{S}_{\text{red}} = \mathbb{R}^7 \times \mathbb{R}^2 \times \mathbb{R}^2$ denote the reduced state space, with elements denoted by $\Theta = (x,r,p)$. Let the terminal condition be denoted by \[g(\Theta; d) = \left(\frac{x_1 - x_2 + K^p x_4 + K^q x_6 + K^0 P_0^*}{K^p x_5 + K^q x_7 + K^0} - x_3\right)^2 - d(x_1 + x_2).\]

\begin{enumerate}
    \item An upper semicontinuous function $u: [0,T] \times \mathcal{S}_{\text{red}} \to \mathbb{R}$ is a viscosity subsolution to \eqref{eq:zsolution:reduce} if $u(T,\Theta) \leq g(\Theta; d)$ for all $\Theta \in \mathcal{S}_{\text{red}}$, and for any $(t_0, \Theta_0) \in [0,T) \times \mathcal{S}_{\text{red}}$ and any test function $\psi \in \mathcal{C}^{1,2}([0,T] \times \mathcal{S}_{\text{red}})$ with polynomial growth such that $u - \psi$ attains a global maximum at $(t_0, \Theta_0)$ with $u(t_0, \Theta_0) = \psi(t_0, \Theta_0)$, we have:
    \begin{equation*}
        -\frac{\partial \psi}{\partial t}(t_0, \Theta_0) - \inf_{z\in \mathbb{R}^{7\times 2}} \tilde{\mathcal L}\psi(t_0, \Theta_0; z, d) \leq 0.
    \end{equation*}

    \item A lower semicontinuous function $v: [0,T] \times \mathcal{S}_{\text{red}} \to \mathbb{R}$ is a viscosity supersolution to \eqref{eq:zsolution:reduce} if $v(T,\Theta) \geq g(\Theta; d)$ for all $\Theta \in \mathcal{S}_{\text{red}}$, and for any $(t_0, \Theta_0) \in [0,T) \times \mathcal{S}_{\text{red}}$ and any test function $\psi \in \mathcal{C}^{1,2}([0,T] \times \mathcal{S}_{\text{red}})$ with polynomial growth such that $v - \psi$ attains a global minimum at $(t_0, \Theta_0)$ with $v(t_0, \Theta_0) = \psi(t_0, \Theta_0)$, we have:
    \begin{equation*}
        -\frac{\partial \psi}{\partial t}(t_0, \Theta_0) - \inf_{z\in \mathbb{R}^{7\times 2}} \tilde{\mathcal L}^z\psi(t_0, \Theta_0; z, d) \geq 0.
    \end{equation*}

    \item A continuous function $\varphi: [0,T] \times \mathcal{S}_{\text{red}} \to \mathbb{R}$ is a viscosity solution to \eqref{eq:zsolution:reduce} if it is both a viscosity subsolution and a viscosity supersolution.
\end{enumerate}
\end{definition}
Motivated by numerical schemes, we will restrict our study of this HJB-PDE on a bounded domain of $\mathcal{S}_{\text{red}} = \mathbb{R}^7 \times \mathbb{R}^2 \times \mathbb{R}^2$ denoted by $\mathcal S_m$, where each variables are bounded uniformly by a constant $m>0$.

\begin{lemma}
    For any $\Theta\in \mathcal S_m, (\mu,\lambda)\in B\times \Lambda$ and bounded fees $d\in [0,d_\infty]$, the functions $g, F^i$ are uniformly Lipschitz in $\Theta,d$.
\end{lemma}
\begin{proof}
    The proof is a direct consequence of the boundedness of $\mathcal S_m,B,\Lambda$ together with Remark \ref{rem:existencecp}.
\end{proof}

In order to guarantee the existence and uniqueness of a viscosity solution to \eqref{eq:zsolution:reduce}, we need to enforce the following assumption.\newline

\textbf{(L)} The optimizers $\hat\alpha^p,\hat\alpha^q$ given by Theorem \ref{thm:main} are uniformly Lipschitz in $\Theta,d$.\newline

This assumption is in particular satisfies for symmetric Auction Makers, with the symmetric Nash equilibrium $\hat\mu^p=\hat\mu^q=\frac{\ln(\gamma)}{c}$, see Remark \ref{rem:Ucharacterization} and Corollary \ref{cor:Nash}. As a direct consequence together with for exmaple \cite{jing2013regularity} we have the following theorem.

\begin{theorem}[Existence and Uniqueness of viscosity solutions]
    Under Assumption \textbf{(L)} there exists a unique viscosity solution to PDE \eqref{eq:zsolution:reduce}. 
\end{theorem}

Following \cite{dayanikli2025machine}, we use a neural network to approximate the function $Z(t,x,y,r,p;d)$, defining the network’s loss function as the stock exchange’s objective  in \eqref{eq:exchange_incentive} and \eqref{eq:ex_const}. The reason a neural network can be used to learn a function is supported by the Universal Approximation Theorem \cite{hornik1989multilayer}, which states that a neural network with sufficient layers and appropriate activation functions can approximate any $C^k$ functions effectively.


\subsection{Algorithm and Parameter Calibration}

For a given time interval $[0, T]$, we partition it into $n$ sub-intervals: $0 = t_0, t_1, \dots, t_n = T$ with $\Delta t = \frac{T}{n}$. Let the neural network model's input be $(r, x, y, p)$ and the output be $\varphi(t, x, y, p)$. Denote the sample size (batch size) to be $M$. Recall the problem of the exchange in \eqref{opt:tilderho}, \eqref{eq:ex_const}, and \eqref{eq:exchange_incentive}, we define the  network’s objective as: $\min_{d,Z} \bigr\{\rho(Y_T^{0,Z,d},d,\hat\alpha(Y_T^{0,Z,d},d)) \bigr\}$, for some constant $\epsilon>0$. Using the definition of $V_0$ in \eqref{eq:traderproblem}, we can further format this objective as \[\min_{d,Z} \bigr\{\mathbb{E}^{\hat \alpha}\big[|P^{cl}_T - P_T^{*}|^2 - d(X_T^1 + X_T^2)+ Y_T^{0,Z,d} \cdot I_2\big]  \bigr\}.\]

The algorithm \ref{alg:neural} for computing the loss function is provided below:

\begin{algorithm} 
\caption{Loss Function of the Neural Network}
\begin{algorithmic}[1] 
    \STATE Initialize $X_0, Y_0, P_0$
    \FOR{$t = 0$ to $n-1$}
        \STATE Compute $Z_t = 20\times \varphi(t,X_t,Y_t,P_t)$
        \STATE Find a Nash equilibrium $\hat \alpha_t$
        \STATE Compute $X_{t+1}, Y_{t+1}, P_{t+1}$ from $X_{t}, Y_{t}, P_{t}$  
    \ENDFOR
    \STATE Compute $P_T^{cl}$, $P_T^{*}$, $g^p(X_T)$, and $g^q(X_T)$ from $X_{T}, Y_{T}, P_{T}$
    \STATE Loss $=  \Big[\frac{1}{M}\sum_{k=1}^M \Big(|P^{cl}_T - P_T^{*}|^2 + Y_T^{p,Z,d} + Y_T^{q,Z,d} - d(X_T^1 + X_T^2)\Big)\Big] $
    \RETURN Loss
\end{algorithmic}\label{alg:neural}
\end{algorithm}

\begin{remark}
    Note that we multiply $20$ at $Z_t = 20\times \varphi(t,X_t,Y_t,P_t)$ to ensure that \eqref{eq:z_condition} is satisfied, thus a Nash Equilibrium control $(\alpha^p, \alpha^q)$ can be found.
\end{remark}

As for the parameters, we set $T = 10$. This value was chosen randomly but can be interpreted as the auction remaining open for 10 seconds or 10 minutes. Recall that non-strategic auction participants cancel their orders based on the values of $1_{A_i \leq \theta(T - \tau^a_i)}$ and $\mathbf 1_{B_i \leq \theta(T - \tau^b_i)}$; we take $\theta(T - \tau_i) = \frac{1}{1+T-\tau_i} + 0.5$.\vspace{0.5em}

We set $K = K_0 = K^p = K^q = 1$ assuming a \$1 spread between the clearing price and the order price would prompt traders to demand an additional share. We set $v^a =v^b =1$, $\lambda_0 = 100$ assuming a non-strategic trader submits an order for one share and the anchor arrival intensity is 100 orders/traders per unit of time. We set $c = 0.1$, considering the shape of $e^{cx}$ and selecting a value that reasonably reflects how sensitive the arrival intensity of non-strategic traders is to the spread offered by strategic traders. For the reservation parameter, we set $R_0 = [100- 35000,100- 35000]$. The \$100 says each strategic trader gains at least \$100 from the auction whereas $35000$ accounts for the distance between the strategic trader’s $\lambda$ and $\lambda_0$, thereby providing greater flexibility for the strategic trader.\vspace{0.5em}

We calibrate $P_0^*$ and $\sigma$ by using the trading data of ''Apple'' and ''Alphabet'' extracted from YahooFinance on period Oct-2-2023 to Dec-29-2023. The calibration method is illustrated in \cite[ Section 2.4]{mastrolia2024clearing}. We test two sets of parameters: we set $P_0^* = 184.39$, $\sigma = 1.76$ for Apple, and $P_0^* = 134.24$, $\sigma = 2.11$ for Alphabet. We set $\lambda_{\infty} = 200$ and $\mu_{\infty} = 60$; both are large enough such that the program would not touch the bound when searching for the optimal controls. We vary $\epsilon \in [0.5,3]$.  \vspace{0.5em}

As discussed in Section \ref{sec:interpretresult}, $NE(\xi,d)$ may contain multiple pairs of optimal actions $\alpha = (\alpha^p, \alpha^q)$. More specifically, if there exists a positive $\gamma$ that satisfies equation \eqref{eq:gamma_nash}, then for any $(\mu^p, \mu^q)$ such that $\mu^p + \mu^q = \frac{\ln \gamma}{c}$, we have  $(\mu^p, \mu^q) \in NE(\xi,d) $. To select a unique pair $(\mu^p, \mu^q)$ for our algorithm, we fix $(\mu^p, \mu^q) = (0.5*\frac{\ln \gamma}{c}, 0.5*\frac{\ln \gamma}{c})$. Based on our initial tests which selected the pair that minimizes the loss of the exchange, recalling the definition of $NEI(\xi,d)$, from $\frac{\ln \gamma}{c}* (-1, 2)$, $\frac{\ln \gamma}{c}* (-0.5, 1.5)$, $\frac{\ln \gamma}{c}* (0, 1)$, $\frac{\ln \gamma}{c}* (0.5, 0.5)$, and $\frac{\ln \gamma}{c}* (1, 0)$, we found that the pair $(0.5*\frac{\ln \gamma}{c}, 0.5*\frac{\ln \gamma}{c})$ is consistently selected. Thus in consideration of the computation cost we fix $(\mu^p, \mu^q) = (0.5*\frac{\ln \gamma}{c}, 0.5*\frac{\ln \gamma}{c})$.  \vspace{0.5em}

Based on the choice of $(\mu^p, \mu^q)$ discussed above and the discussion of symmetric strategic traders in Remark \ref{rem:Ucharacterization}, we further simplify the numerical method to assume that the exchange sends identical incentive contract to strategic trader $p$ and strategic trader $q$, i.e. $Z^{1,p} = Z^{1,q}$, $Z^{2,p} = Z^{2,q}$, $Z^{3,p} = Z^{3,q}$, $Z^{4,p} = Z^{6,q}$, $Z^{5,p} = Z^{7,q}$, $Z^{6,p} = Z^{4,q}$, $Z^{7,p} = Z^{5,q}$. \vspace{0.5em}

For the numerical method, we set $n =50$, $M = 256$ to balance precision and computation cost. The model uses $2$ hidden layers, with $8$ neurons in the first hidden layer and $7$ neurons in the second hidden layer. Note that the number of neurons in the second hidden layer corresponds to the output dimension. Activation function is set to be the sigmoid function for all layers (including the output layer) except the input layer. We tested other activation functions, such as the identity function, ReLU, and tanh. However, the sigmoid function proved to be the most stable, as it supports the existence of Nash equilibrium controls. We employ Adam optimization with automatic differentiation, using a learning rate chosen from $\{0.001,0.0005,0.0001,0.00005,0.00001\}$. \vspace{0.5em}

\subsection{Illustration and Policy Recommendations}

In terms of the optimal $d$ which minimizes $\rho(d)$, we find $\hat d = 3.0$ for Apple and $\hat d = 2.0$ for Alphabet. Figure \ref{fig:result1} shows that as the transaction fee $d$ initially increases, $\rho(d)$ decreases, likely due to increased fee revenues. However, beyond the optimal point, further increases in $d$ transaction fee cease to be beneficial probably because higher fees discourage non-strategic traders from participating in the auction, ultimately reducing the exchange's overall gains.\vspace{0.5em}

\begin{figure}[htbp]
    \begin{minipage}{0.45\textwidth}
        \centering
        \includegraphics[width=\linewidth]{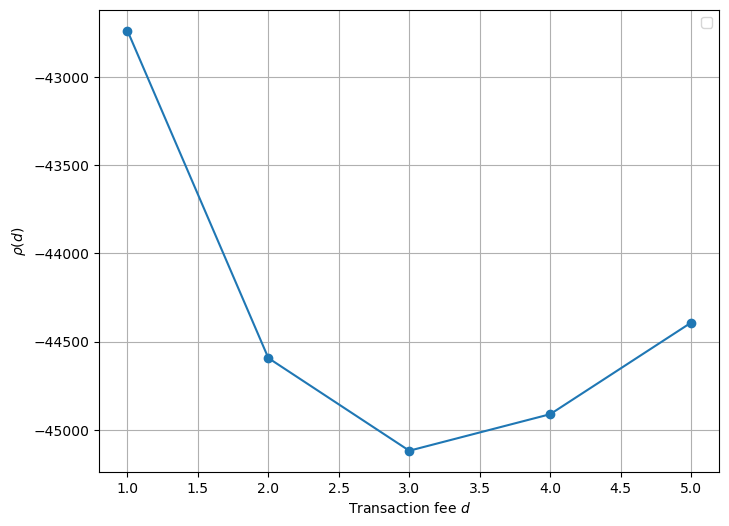}
        \caption*{(a) Apple}
    \end{minipage}
    \begin{minipage}{0.45\textwidth}
        \centering
        \includegraphics[width=\linewidth]{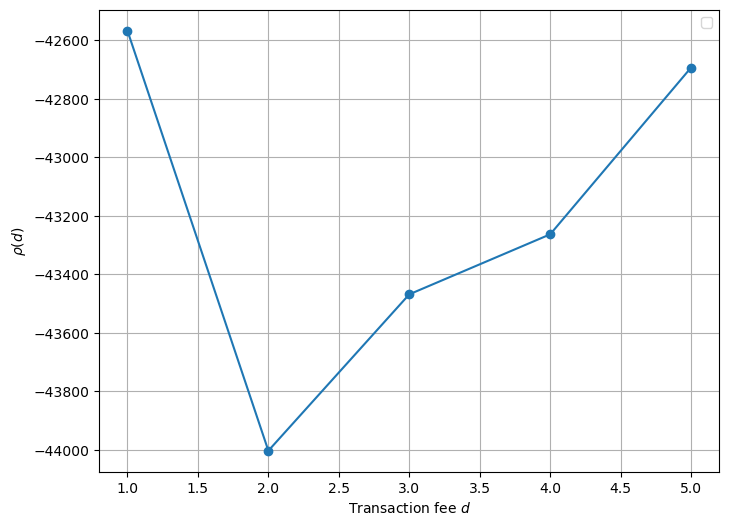}
        \caption*{(b) Alphabet}
    \end{minipage}
    \caption{Clearing-efficient spread $\rho(d)$ as a function of the transaction fee $d$}
    \label{fig:result1}
\end{figure}

We then want to illustrate $\hat \xi(\hat d)=Y_T^{R_0,Z}$. Due to similarity between Apple and Alphabet's result, we only present Apple's result here.\vspace{0.5em}

Figure \ref{fig:result2}(a) illustrates the optimal compensation indexes by the different state variables. We first observe that all these compensations are negative, advocating for the snipping monitoring and incentivizing the strategic actors to arrive as early as possible. Then, $Z^1_t$, $Z^2_t$ remains relatively stable over time $t$, so the exchange can treat these two variables as a constant vector across $t$ for simplicity. Our primary focus is on $Z^4_t$ and $Z^6_t$, which influence the traders' arrival intensities over time. A specific feature of an auction market in contrast to a continuous trading session is that orders are not executed immediately. This delay allows traders to select their arrival time or intensity during the auction to optimize their gains. For instance, traders may choose to arrive late to gain more market information to profit from the early-arriving orders. Early arriving traders can also benefit from the market by placing misleading orders or setting a tone at the beginning of the market. Consequently, the component related to arrival intensity of the incentive becomes important as it can help mitigate the timing advantages present in an auction market. We now discuss the shapes of $Z^4_t$ and $Z^6_t$. From Figure \ref{fig:result2}(a), we see that they increase over a brief period and then gradually decrease over time. This indicates that the exchange does not want the traders to arrive too early but also seeks to avoid overly late arrivals. Due to this incentive structure, we see that the strategic traders reduce their arrival intensity over time, sending fewer orders as the auction progresses. This aligns with \cite{mastrolia2024clearing}, which indicate that the exchange prefers traders to arrive earlier in the auction rather than later and preventing snipping behavior. \vspace{0.5em}

Figures \ref{fig:result2} (b) and (c) show the strategic traders' responses to these incentives, while Figure \ref{fig:result2} (d) shows the arriving intensities of the non-strategic buyers and sellers. It is important to note that the difference between $\lambda^p$ and $\lambda^q$, and between the continuous incentive $F^i(X_s,Z_s,r_s,P_s, \hat \alpha^j_s,d;\hat \alpha^i_s)$, shown in Figure \ref{fig:result2} (b), (e), (f) and (g), are numerical artifact. Theoretically, these values should be identical; however due to the imperfection in generating Brownian motions using a computer program, discrepancies arise. In particular, the value $Y_0$ provides a rebate to the strategic trader while the negativity of the $Z$ components mitigate any snipping effect, as we were expected. \vspace{0.5em}
Figure \ref{fig:result2} (e) displays the average optimal $-\int_0^t F^i(X_s,Z_s,r_s,P_s, \hat \alpha^j_s,d;\hat \alpha^i_s) ds$, where the blue solid line represents the continuous incentive for trader $i=p$, and the orange dashed line represents the incentive for trader $i=q$.\vspace{0.5em}

Figure \ref{fig:result2} (f) displays the average optimal $F^i(X_s,Z_s,r_s,P_s, \hat \alpha^j_s,d;\hat \alpha^i_s)$. Figure \ref{fig:result2} (g) displays the total rebate $Y^{0,Z,\hat d}_t$, with $\xi (\hat d) = Y^{0,Z,\hat d}_t$. Figure \ref{fig:result2} (a) shows $Z^{1},Z^{2}, ..., Z^{7}$. 



\begin{figure}[htbp]
    \centering
    \begin{minipage}{0.43\textwidth}
        \centering
        \includegraphics[width=\linewidth]{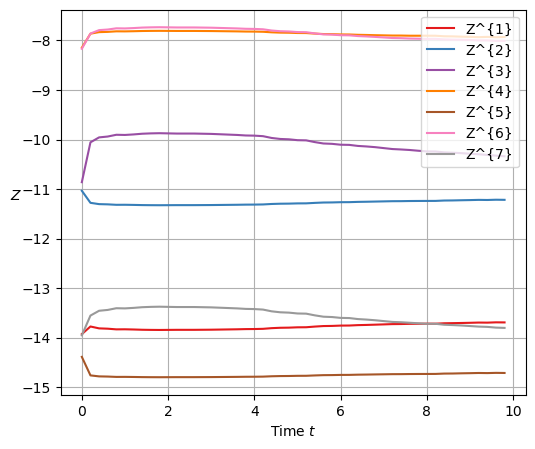}
        \caption*{(a) Selected $Z$}
    \end{minipage}
    \begin{minipage}{0.43\textwidth}
        \centering
        \includegraphics[width=\linewidth]{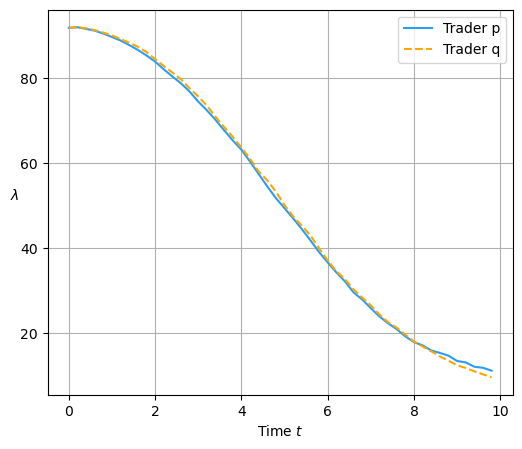}
        \caption*{(b) Traders' optimal arriving intensities, $\hat{\lambda}$}
    \end{minipage}
    
    \begin{minipage}{0.43\textwidth}
        \centering
        \includegraphics[width=\linewidth]{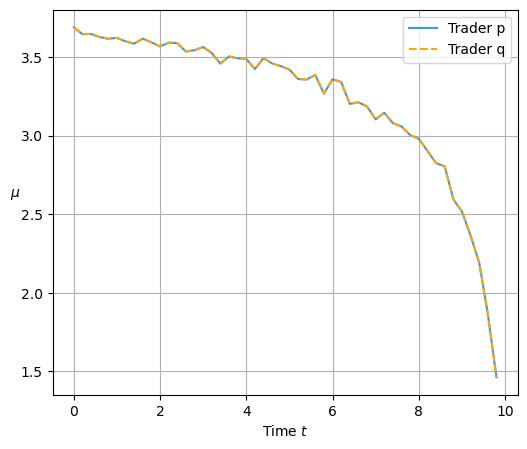}
        \caption*{(c) Traders' optimal spreads, $\hat{\mu}$}
    \end{minipage}
    \begin{minipage}{0.43\textwidth}
        \centering
        \includegraphics[width=\linewidth]{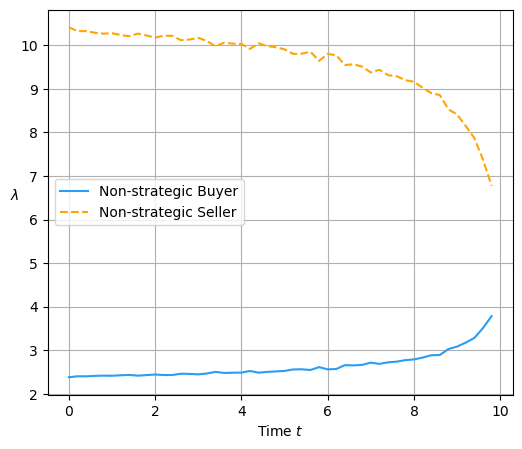}
        \caption*{(d) The value of $\lambda^a$ and $\lambda^b$ }
    \end{minipage}
    
    \begin{minipage}{0.45\textwidth}
        \centering
         \includegraphics[width=\linewidth]{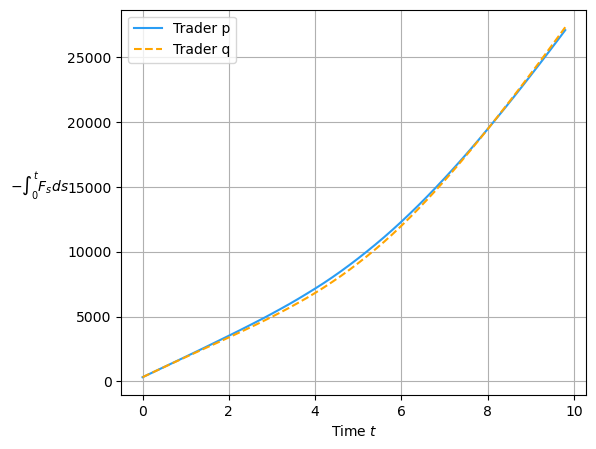}
        \caption*{(e) Exchange's optimal incentive -$\int_0^t F_s ds$ }
    \end{minipage}
    \begin{minipage}{0.45\textwidth}
        \centering
         \includegraphics[width=\linewidth]{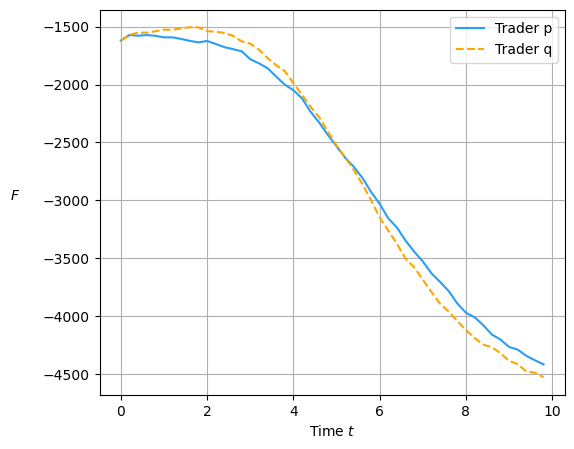}
        \caption*{(f) Exchange's optimal $F$ }
    \end{minipage}
\end{figure}

\begin{figure}[htbp]
    \centering
    \begin{minipage}{0.45\textwidth}
        \centering
             \includegraphics[width=\linewidth]{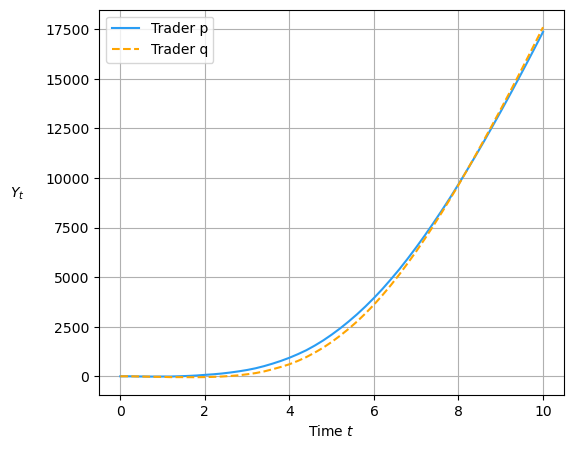}
        \caption*{(g) Continuation utilities $Y^p_t,Y_t^q$. }
    \end{minipage}

    \caption{Result for Apple. Each graph shows the  mean of 256 sample paths.}
    \label{fig:result2}
\end{figure}

Additionally, we see that the fee is effective in regulating other behaviors of strategic and non-strategic traders throughout the auction. The strategic traders' spreads, $\mu^p$ and $\mu^q$, decrease to zero as time advances. This indicates that the strategic traders are submitting orders that align more closely with the efficient price, rather than introducing misleading or disruptive pricing signals. Furthermore, the imbalance between the arriving intensity of non-strategic buyers and non-strategic sellers diminishes due to the closing gap between the strategic traders' price and the efficient price. With regard to the incentive structure, we observe from figure \ref{fig:result2} (e) and (f) that the time-dependent component of the incentives and its derivatives increase over time. Figure \ref{fig:result2} (g) indicates that the total incentive also increases over time. This pattern may be attributed to strategic traders becoming less inclined to participate and earning less as the auction progresses, so the exchange compensates them by enhancing the incentives. The insight we gain here can be a potential strategy for the exchange to design incentives. Specifically, an effective incentive could be formulated as an inverse function relative to the traders' gains or the traders' willingness to participate. By shaping the incentive structure in this way, the exchange can effectively regulate the auction. \vspace{0.5em}


Lastly, we demonstrate the effectiveness of the fee and incentive in improving both market efficiency and the exchange's outcomes by comparing two scenarios: one with no fees or incentives and one with the optimal fee and incentive $(\hat{d}, \hat{\xi}(\hat{d}))$. When $(d = 0, \xi = 0)$, the exchange's value is $\rho = 454$ for Apple and $\rho = 487$ for Alphabet, whereas with $(\hat{d}, \hat{\xi}(\hat{d}))$, it improves significantly to $\rho(\hat{d}) = -45119$ for Apple and $\rho(\hat{d}) = -44003$ for Alphabet. This demonstrates that the fee and incentive substantially enhance the exchange’s outcomes. In terms of market efficiency, when $(d = 0, \xi = 0)$, the expected market spread is $\mathbb{E}^{\alpha(0,0)}\Big[|P^{cl}_T - P_T^{}|^2\Big] = 454$ for Apple and $487$ for Alphabet, compared to $\mathbb{E}^{\alpha(\hat{\xi}, \hat{d})}\Big[|P^{cl}_T - P_T^{}|^2\Big] = 144$ for Apple and $127$ for Alphabet under the optimal fee and incentive structure. Therefore, we recommend implementing the fee and the incentive to promote both the welfare of the exchange and the overall efficiency of the market.


%
%


  \newpage 
  \small
  \bibliographystyle{apalike}
\bibliography{ref}

\appendix

\section{Proof of Lemma \ref{thm:girsanoc}}\label{app:girsanov}

\begin{proof}
\textit{Proof of i.} This is a direct consequence of \cite[Proposition 2.10]{bjork2021point} or \cite[Theorem 1.36]{oksendal2019stochastic}.

\textit{Proof of ii.} By Definition 4.1 in \cite{bjork2021point}, it is sufficient to show that for every non-negative predictable process $h$ and every $i \in \{1, \dots, k\}$, $\mathbb{E}^\mathbb{Q}\left[ \int_0^T h_t\lambda_t^i(1+\phi^i_t) dt \right] = \mathbb{E}^\mathbb{Q}\left[ \int_0^T h_t dN^i_t \right]$.

Define $Z_t := L_t Y_t$, where $Y_t := \int_0^t h_s dN^i_s$. By Itô's formula for semimartingales with jumps, 
\begin{align*}
    Z_t &= Z_0 + \int_0^t Y_{s^-} dL_s + \int_0^t L_{s^-} dY_s + [L, Y]_t \\
    &= \int_0^t Y_{s^-} dL_s + \int_0^t L_{s^-} h_s dN^i_s + \sum_{0<s\leq t} \Delta L_s \Delta Y_s \\
    &= \int_0^t Y_{s^-} dL_s + \int_0^t L_{s^-} h_s dN^i_s + \sum_{0<s\leq t} \left(L_{s^-} \sum_{j=1}^k \phi_s^j \Delta N_s^j\right) (h_s \Delta N^i_s).
\end{align*}

By part i, $L$ is a martingale, so $\mathbb{E}^\mathbb{P} \big[\int_0^T Y_{s^-} dL_s\big] = 0$. By the mutual independence of $N^1, \dots, N^k$ under $\mathbb{P}$, we have $\mathbb{E}^\mathbb{P} \big[\sum_{0<s\leq t} \Delta N_s^j \Delta N^i_s\big] = 0$ for $i \ne j$.
Taking the expectation under $\mathbb{P}$ at terminal time $T$ yields:
\begin{align*}
    \mathbb{E}^\mathbb{Q} \left[ \int_0^T h_t dN^i_t \right] &= \mathbb{E}^\mathbb{P} [Z_T] \\
    &= \mathbb{E}^\mathbb{P} \left[ \int_0^T L_{s^-} h_s dN^i_s + \sum_{0<s\leq T} L_{s^-}\phi_s^i h_s (\Delta N_s^i)^2 \right] \\
    &= \mathbb{E}^\mathbb{P} \left[ \int_0^T L_{s^-} h_s (1 +\phi_s^i) dN^i_s \right].
\end{align*}
Since the compensator of $N^i$ under $\mathbb{P}$ is $\lambda^i_s ds$, the predictable projection allows us to write this as $\mathbb{E}^\mathbb{P} \left[ \int_0^T L_{s^-} h_s (1 +\phi_s^i)\lambda^i_s ds \right]$, which precisely equals $\mathbb{E}^\mathbb{Q} \left[ \int_0^T h_s\lambda_s^i(1+\phi^i_s) ds \right]$.

\textit{Proof of iii.} Note that $[\tilde{W}^i,\tilde{W}^i]_t = [W^i,W^i]_t = t$. By Lévy's characterization theorem \cite[Theorem 4.6.4]{shreve2004stochastic}, it is sufficient to show that $\tilde{W}^i$ is a $\mathbb{Q}$-martingale. 
Fix $0 \leq s < t \leq T$. By Bayes' rule for conditional expectations \cite[Lemma 5.2.2]{shreve2004stochastic}, $\mathbb{E}^\mathbb{Q}[\tilde{W}^i_t|\mathcal{F}_s] = \frac{1}{L_s}\mathbb{E}^\mathbb{P}[\tilde{W}^i_t L_t|\mathcal{F}_s]$. To show that this equals $\tilde{W}^i_s$, it is sufficient to show that $\tilde{W}^i L$ is a $\mathbb{P}$-martingale.

By Itô's product rule, 
\begin{align*}
    \tilde{W}^i_t L_t &= \tilde{W}^i_s L_s + \int_s^t \tilde{W}^i_{v^-} dL_v + \int_s^t L_{v^-} d\tilde{W}^i_v + [\tilde{W}^i, L]_t - [\tilde{W}^i, L]_s \\
    &= \tilde{W}^i_s L_s + \int_s^t \tilde{W}^i_{v^-} \left(L_v \sum_{j=1}^\ell \delta_v^j dW_v^j + L_{v^-} \sum_{j=1}^k \phi_v^j (dN_v^j - \lambda_v^j dv)\right) \\
    &\quad + \int_s^t L_v (dW_v^i - \delta_v^i dv) + \int_s^t L_v \delta_v^i dv \\
    &= \tilde{W}^i_s L_s + \int_s^t \left(L_v + \tilde{W}^i_{v^-} L_v \delta_v^i \right) dW_v^i + \int_s^t \tilde{W}^i_{v^-} L_v \sum_{j \neq i}^\ell \delta_v^j dW_v^j \\
    &\quad + \int_s^t \tilde{W}^i_{v^-} L_{v^-} \sum_{j=1}^k \phi_v^j (dN_v^j - \lambda_v^j dv).
\end{align*}
Since the integrals on the right-hand side are taken with respect to $\mathbb{P}$-martingales (the Brownian motions $W$ and the compensated Poisson processes), taking the conditional expectation yields $\mathbb{E}^\mathbb{P}[\tilde{W}^i_t L_t|\mathcal{F}_s] = \tilde{W}^i_s L_s$.

\textit{Proof of iv.} Fix $t \in [0,T]$ and $i \ne j \in \{1, \dots, k\}$. We claim $\mathbb{Q}(\Delta N_t^i \Delta N_t^j \neq 0) = 0$. By the independence of $N^i$ and $N^j$ under the reference measure $\mathbb{P}$, they have no simultaneous jumps a.s., meaning $\mathbb{P}(\Delta N_t^i \Delta N_t^j \neq 0) = 0$ \cite[Proposition 3.1]{bjork2021point}. Since $\mathbb{Q}$ is absolutely continuous with respect to $\mathbb{P}$ ($\mathbb{Q} \ll \mathbb{P}$), it immediately follows that $\mathbb{Q}(\Delta N_t^i \Delta N_t^j \neq 0) = 0$.
\end{proof}

In our auction model, we consider two market makers, $p$ and $q$, who control their respective bidding prices through deviations from the efficient price (controlled via the drifts $\mu^p, \mu^q$) and actively control their trading intensities ($\lambda^p, \lambda^q$). We denote by $\mathcal{A}$ the set of admissible control pairs $\alpha=(\alpha^p,\alpha^q)$ where $\alpha^i=(\mu^i,\lambda^i)$ for $i \in \{p,q\}$. We assume $\mu^i$ is a $B$-valued $\mathbb{F}$-predictable process and $\lambda^i$ is a $\Lambda$-valued $\mathbb{F}$-predictable process, where $B \subset \mathbb{R}$ and $\Lambda \subset \mathbb{R}_+$ are bounded sets. We denote by $\mathcal{A}^i$ the set of admissible strategies for player $i$. 

Consequently, the pair $(\mu, \lambda)$ satisfies the integrability condition \eqref{novikov} by setting $\delta^i = \frac{\mu^i}{\sigma}$ and $\phi^i = \frac{\lambda^i_t - \lambda_0}{\lambda_0}$. The Radon-Nikodym density process $L^\alpha$ is then the solution to the SDE:
\begin{align*}
    dL_t^{\alpha} = L_t^{\alpha}\sum_{i \in \{p,q\}} \frac{\mu_t^i}{\sigma}dW_t^i + L_{t^-}^{\alpha}\sum_{i \in \{p,q,a,b\}} \frac{\lambda^i_t - \lambda_0}{\lambda_0} (dN_t^i - \lambda_0 dt), \quad L^\alpha_0 = 1.
\end{align*}
This stochastic exponential has the explicit solution \cite[Theorem 11.6.10]{shreve2004stochastic}:
\begin{align*}
    L_t^{\alpha} = \exp \Bigg( &\sum_{i \in \{p,q\}} \int_0^t \left[ \frac{\mu_s^i}{\sigma} dW^i_s - \frac{(\mu_s^i)^2}{2\sigma^2}ds \right] \\
    &+ \sum_{i \in \{a,b,p,q\}} \int_0^t \left[ \ln\left(\frac{\lambda^i_s}{\lambda_0}\right) dN_s^i + (\lambda_0 - \lambda^i_s)ds \right] \Bigg), \quad t \in [0,T].
\end{align*}
Because the controls $\lambda$ and $\mu$ are bounded, $L^\alpha$ satisfies standard integrability criteria and admits moments of all orders.

\section{Proof of Lemma \ref{lemma:trader}}\label{app:bestreact} 
The proof is divided into four main steps. Step (a) shows that $g^p(X_t)$ is $\mathbb{P}^\alpha$-integrable. Step (b) derives the Itô decomposition of the dynamic value function, proving $\mathcal{C}^{R_0^p}_p = \Xi^p$. Step (c) characterizes the optimal control $\alpha^p = (\hat{\mu}, \hat{\lambda})$, proving the second assertion. Step (d) establishes the uniqueness of the representation of a rebate in $\Xi^p$. 

\textit{Step (a)}: Let $\xi^p \in \mathcal{C}^{R_0^p}_p$, $(\alpha^p,\alpha^{q}) \in \mathcal{A}$, and $d \geq 0$. We define the dynamic value process of Auction Maker $p$ as:
\begin{align*}
    V_t(\xi^p,d,\alpha^{q}; \alpha^p) := \mathbb{E}^{\alpha} \left[ g^p(X_T) - g^p(X_t) + \xi^p - \int_t^T \frac{(\lambda^p_s - \lambda_0 )^2}{2} ds \;\middle|\; \mathcal{F}_t \right].
\end{align*} 
Notice that $V_T = \xi^p$ and $V_0 = \mathbb{E}^{\alpha} \big[ g^p(X_T) + \xi^p - \int_0^T \frac{(\lambda^p_s - \lambda_0 )^2}{2} ds \big]$ since $g^p(X_0)=0$. To show that $V_t$ is well-defined, note that $\xi^p \in L^1(\mathbb{P}^{\alpha})$ and $\lambda^p$ is bounded. Thus, it suffices to show that $g^p(X_t) \in L^1(\mathbb{P}^{\alpha})$ for every $t \in [0,T]$. 

Let $\tilde{W}^{p}_t = W^{p}_t - \int_0^t \frac{\mu^p_s}{\sigma} ds$. Under $\mathbb{P}^{\alpha}$, $\tilde{W}^{p}$ is a standard Brownian motion. Define $W^{p}_{\max} := \max_{t \in [0,T]} \tilde{W}^{p}_t$ and $W^{p}_{\min} := \min_{t \in [0,T]} \tilde{W}^{p}_t$. By the reflection principle of Brownian motion, $W^{p}_{\max}$ has the same distribution as $|\tilde{W}^{p}_T|$, and $W^{p}_{\min}$ has the same distribution as $-W^{p}_{\max}$. Additionally, $\big|\int_0^t K^p P^p_s dN^p_s\big| \leq K^p N^p_t \check{P}_t$, where $\check{P}_t := |P_0^*| + t \mu_{\infty} + \sigma W^{p}_{\max} - \sigma W^{p}_{\min}$. We define $W^{q}_{\max}$ and $W^{q}_{\min}$ similarly for Auction Maker $q$. For every $t \in [0,T]$, by the triangle and Cauchy-Schwarz inequalities:
\begin{align*}
    \mathbb{E}^{\alpha} \big[|g^p(X_t)|\big] &= \mathbb{E}^{\alpha} \left[ \Big|K^p(P^{cl}_t - P^{*}_t) \Big(P^{cl}_t N^p_t - \int_0^t P^p_{s}dN^p_s\Big) \Big| \right] \\
    &\leq \mathbb{E}^{\alpha} \big[K^p|P^{cl}_t|^2 N^p_t\big] + \mathbb{E}^{\alpha} \big[ K^p |P^{cl}_t| N^p_t \check{P}_t \big] \\
    &\quad + \mathbb{E}^{\alpha} \big[ K^p |P^{*}_t P^{cl}_t| N^p_t \big] + \mathbb{E}^{\alpha} \big[K^p |P^{*}_t| N^p_t \check{P}_t\big] \\
    &\leq \mathbb{E}^{\alpha} \big[K^p|P^{cl}_t|^2 N^p_t\big] + \sqrt{\mathbb{E}^{\alpha} \big[(K^p P^{cl}_t N^p_t)^2\big]} \left(\sqrt{\mathbb{E}^{\alpha} \big[(\check{P}_t)^2\big]} + \sqrt{\mathbb{E}^{\alpha} \big[(P_t^*)^2\big]}\right) \\
    &\quad + \sqrt{\mathbb{E}^{\alpha} \big[(K^p N^p_t)^2\big] \mathbb{E}^{\alpha} \big[(P_t^* \check{P}_t)^2 \big]}.
\end{align*}
Since the arrival intensities $\lambda_t^a, \lambda_t^b, \lambda^p_t, \lambda_t^{q}$ are uniformly bounded by their respective maxima, the second moments of the counting processes $N_t^a, N_t^b, N^p_t, N^q_t$ are uniformly bounded in $t$. Furthermore, the moments of the running maxima of Brownian motion are finite (e.g., $\mathbb{E}^{\alpha} [(W^{p}_{\max})^2] = \sigma^2 T$). Expanding $(P^{cl}_t N^p_t)^2$ using the explicit definition of the clearing price \eqref{eq:pclearing} reveals it is bounded by a polynomial of bounded random variables. We conclude that $\mathbb{E}^{\alpha} [|g^p(X_t)|] < \infty$.

\textit{Step (b)}: By Itô's formula for semimartingales with jumps (and noting $\frac{\partial^2 g^p}{\partial x_3^2} = 0$ since $g^p$ is linear in the reference price $X^3$), we have:
\begin{align*}
    g^p(X_T) &= \int_0^T \frac{\partial g^p}{\partial x_3}(X_{s^-}) \sigma dW_s + \int_0^T \big[ g^p(X_{s^-} + r_s^1\mathbf{e}_1) - g^p(X_{s^-}) \big] dN^{a}_s \\
    &\quad + \int_0^T \big[ g^p(X_{s^-} + r_s^2\mathbf{e}_2) - g^p(X_{s^-}) \big] dN^{b}_s + \int_0^T \big[ g^p(X_{s^-} + P^p_s\mathbf{e}_4 + \mathbf{e}_5) - g^p(X_{s^-}) \big] dN^p_s \\
    &\quad + \int_0^T \big[ g^p(X_{s^-} + P^q_s\mathbf{e}_6 + \mathbf{e}_7) - g^p(X_{s^-}) \big] dN^{q}_s.
\end{align*}
Define the martingale:
\begin{equation}
    M_t(\xi^p,\alpha^{q}, \alpha^p,d) := V_t(\xi^p,\alpha^{q}, \alpha^p,d) + g^p(X_t) - \int_0^t \frac{(\lambda^p_s - \lambda_0 )^2}{2} ds.
    \label{eq:Mdef}
\end{equation}
By the Martingale Representation Theorem \cite[Lemma A.1]{euch2021optimal}, there exists a predictable process $\tilde{Z} = (\tilde{Z}^1, \dots, \tilde{Z}^7)$ such that \begin{align*}
    M_t = M_0 &+ \int_0^t \tilde{Z}^1_s (dN^a_s - \lambda^a_s ds) + \int_0^t \tilde{Z}^2_s (dN^b_s - \lambda^b_s ds) + \int_0^t \tilde{Z}^3_s dW_s \\
    &+ \int_0^t \tilde{Z}^4_s (dN^p_s - \lambda^p_s ds) + \int_0^t \tilde{Z}^5_s \left(dW^p_s - \frac{\mu^p_s}{\sigma} ds\right) \\
    &+ \int_0^t \tilde{Z}^6_s (dN^q_s - \lambda^q_s ds) + \int_0^t \tilde{Z}^7_s \left(dW^q_s - \frac{\mu^q_s}{\sigma} ds\right).
\end{align*} 
Equating $M_T$ with its terminal value and substituting the Itô expansion of $g^p(X_T)$, we can isolate $\xi^p$. By identifying the coefficients 
we recover the exact representation $\xi^p = Y_T^{p, Y_0^p, Z^p, d}$ driven by the generator $F^p$.

\textit{Step (c)}: We fix $\xi^p \in \Xi^p$, $\alpha^{q} \in \mathcal{A}^q$, and $d \geq 0$. Define the feedback controls $\hat{\mu}^p_s$ and $\hat{\lambda}^p_s$ as the maximizers of the generator $F^p$. Because the constraints $B$ and $\Lambda$ are bounded, the maximizers $\hat{\alpha}^p = (\hat{\mu}^p, \hat{\lambda}^p)$ are predictable, bounded, and thus admissible in $\mathcal{A}^p$. 
By setting $Y^{\alpha^p}_t = V_t(\xi^p, \alpha^{q}, \alpha^p, d)$, we get

\begin{align}\label{BSDE:control}
\nonumber Y_t^{\alpha^p}&=\xi^p +\int_t^T F^p(X_s,Z_s^p,r_s,P_s,\alpha_s^q,d;\alpha_s^p)ds-\int_t^TZ_s^{1,p}(dN_s^a-\lambda_0ds)\\
\nonumber &-\int_t^T Z_s^{2,p}(dN_s^b-\lambda_0ds)-\int_t^T Z_s^{3,p}dW_s-\int_t^T Z_s^{4,p}(dN_s^p-\lambda_0ds)\\
&-\int_t^T Z_s^{5,p}dW_s^p-\int_t^T Z_s^{6,p}(dN_s^q-\lambda_0 ds)-\int_t^T Z_s^{7,p}dW_s^q.
\end{align}

We define $(\hat Y,\hat Z)$ the solution of the following BSDE

\begin{align*}
\hat {Y}_t&=\xi^p +\int_t^T F^p(X_s,\hat {Z}_s^p,r_s,P_s,\alpha_s^q,d;\hat\alpha_s^p)ds-\int_t^T\hat {Z}_s^{1,p}(dN_s^a-\lambda_0ds)\\
&-\int_t^T \hat {Z}_s^{2,p}(dN_s^b-\lambda_0ds)-\int_t^T \hat {Z}_s^{3,p}dW_s-\int_t^T \hat {Z}_s^{4,p}(dN_s^p-\lambda_0ds)\\
&-\int_t^T \hat {Z}_s^{5,p}dW_s^p-\int_t^T \hat {Z}_s^{6,p}(dN_s^q-\lambda_0 ds)-\int_t^T \hat {Z}_s^{7,p}dW_s^q.
\end{align*}
Note that the existence and uniqueness is satisfied from \cite{papapantoleon2018existence}. From \cite[Theorem 3.25]{papapantoleon2018existence} we deduce that for any control $\alpha^p$

\[Y_t^{\alpha^p}\leq \hat Y_t,\text{ for any }t\in [0,T],\; \mathbb P^0-a.s.\]

Hence, $\hat{\alpha}^p$  maximizes the expected utility of Maker $p$, meaning it is Auction Maker $p$'s optimal response. Conversely, assume that $\alpha^{p,\star}$ is the best response of Auction Maker $p$ solving \eqref{eq:traderproblem} when $\alpha^q$ is fixed. Then, $\hat Y_0=Y_0^{\alpha^{\star,p}}.$ Using a strict comparison theorem, see for example \cite[Theorem 6.1. (b)]{nie2021existence} we deduce that $\alpha^{p,\star}$ is a maximizer of $F^p$. 

\textit{Step (d)}: Suppose there exist two representations $(Y_0, Z)$ and $(\tilde{Y}_0, \tilde{Z})$ for the same terminal contract $\xi^p$.
For any $t \in [0,T]$, the dynamic programming principle implies that $Y_t$ and $\tilde{Y}_t$ both equal the essential supremum of the continuation value given the state at time $t$. Therefore, $Y_t = \tilde{Y}_t$ almost surely for all $t \in [0,T]$. Taking the quadratic variation of their difference implies $Z = \tilde{Z}$ in $\mathbb H^2([0,T])$, proving uniqueness.

\end{document}